\documentclass[3p,times]{article}
\usepackage{tda}

\begin{document}

\title{Reconstructing Embedded Graphs from Persistence Diagrams}

\author{Robin Lynne Belton\thanks{Depart. of Mathematical Sciences,
    Montana State U.}
    \and
        Brittany Terese Fasy\footnotemark[1]~\footnotemark[2]
        \and
        Rostik Mertz\thanks{School of Computing,
            Montana State
            U. \newline
   {\scriptsize {\tt \{robin.belton, brittany.fasy, samuelmicka, david.millman, annaschenfisch, jordan.schupbach, luciawilliams\}@montana.edu}
    {\tt antonmertz@gmail.com}
    {\tt dsalinasduron@westminstercollege.edu}}}
        \and
        Samuel Micka\footnotemark[2]
        \and
        David L. Millman\footnotemark[2]
        \and
        Daniel Salinas\footnotemark[2]
        \and
        Anna Schenfisch\footnotemark[1]
        \and
        Jordan Schupbach\footnotemark[1]
        \and
        Lucia Williams\footnotemark[2]
    }

\index{Belton, Robin Lynne}
\index{Fasy, Brittany Terese}
\index{Mertz, Rostik}
\index{Micka, Samuel}
\index{Millman, David L.}
\index{Salinas, Daniel}
\index{Schenfisch, Anna}
\index{Schupbach, Jordan}
\index{Williams, Lucy}

\maketitle 

\begin{abstract}
    The persistence diagram (PD) is an increasingly popular topological descriptor.
By encoding the size and prominence of topological features at varying scales,
the PD provides important geometric and topological information about a space.
Recent work has shown that well-chosen (finite) sets of PDs can differentiate between
geometric simplicial complexes, providing a method for representing complex
shapes using a finite set of descriptors.
A related inverse problem is the following: given a set of PDs
(or an oracle we can query for persistence diagrams), what is underlying geometric
simplicial complex?  In this paper, we present an
algorithm for reconstructing embedded graphs in $\R^d$ (plane graphs in $\R^2$)
with $n$ vertices
from $n^2 - n + d + 1$ directional (augmented) PDs.
Additionally, we empirically validate the
correctness and time-complexity of our algorithm in $\R^2$ on randomly generated plane graphs using our implementation,
and explain the numerical limitations of implementing our algorithm.

\end{abstract}

\section{Introduction}
\label{sec:intro}

Topological data analysis (TDA) provides a set of promising tools to help
analyze data in fields as varied as
materials science, transcriptomics, and neuroscience \cite{giusti2015clique,lee2017quantifying,rizvi2017single}.
The wide applicability is due to the fact that many forms of data can be
modeled as graphs or simplicial complexes, two widely-studied types of
topological spaces.
Topological spaces are described in terms of their invariants--such as the
homotopy type or homology classes. Persistent homology considers the evolution
of the homology groups in a filtered topological space.

\paragraph{Motivation}
The problem of manifold and stratified space learning is an active research area
in computational mathematics. For example, Chambers et al.\ use persistent
homology in a stratified space setting~\cite{chambers2018heuristics}, describing an algorithm to identify
shapes that simplify a noisier shape, and then confirm that the given simplification
still satisfies the desired topological properties. Zheng et al.\
also address a problem in space learning, and study the 3D reconstruction of
plant roots from multiple 2D images~\cite{rootreconstruction}, using persistent homology to ensure the resulting 3D root model is connected.
Another reconstruction problem involves reconstructing road networks.
Three common approaches to solving these problems involve Point Clustering,
Incremental Track Insertion, and Intersection Linking~\cite{maps}.
Ge, Safa, Belkin, and Wang develop a point clustering algorithm using Reeb
graphs to extract the skeleton graph of a road from point-cloud
data~\cite{ge2011data}. The original embedding can be reconstructed using a
principal curve algorithm~\cite{kegl2000learning}.
Karagiorgou and Pfoser give an algorithm to reconstruct a road network from vehicle
trajectory GPS data by identifying intersections with clustering, then using vehicle
trajectories to connect them~\cite{ili}. Ahmed et  al.\  provide an incremental track insertion
algorithm to reconstruct road networks from point cloud data~\cite{iti}. The
reconstruction is done incrementally, using a variant of the Fr\'echet distance
to partially match input trajectories to the reconstructed graph.
Ahmed, Karagiorgou, Pfoser, and Wenk describe all these methods in~\cite{maps}.
Finally, Dey, Wang, and Wang use persistent homology to reconstruct embedded
graphs. This research has also been applied to input trajectory
data~\cite{dey2018graph}. Dey et al.\ use persistence to guide the Morse
cancellation of critical simplices. We see from these applications the
necessity for reconstruction algorithms, and in particular the necessity for
reconstruction algorithms of graphs since much of the research involving
reconstruction of road networks involves reconstructing graphs.

We explore the reconstruction of graphs from a widely used topological
descriptor of data, \emph{persistence diagrams}.
The problem of reconstruction
for simplicial complexes has received
significant recent attention \cite{curry2018directions,
ghrist2018persistent,turner2014persistent}.
Our work is motivated by \cite{turner2014persistent}, which proves that one can
reconstruct simplicial complexes in~$\R^2$ and $\R^3$
using a subset of a parameterized infinite set of persistence diagrams.
In this work, we use a version
of persistence diagrams that includes information that is not normally
considered; thus, for clarity, we refer to these descriptors as
\emph{augmented persistence diagrams} (APDs).  Our approach
differs from those listed above because we provide a deterministic
algorithm using APDs generated from specific
directional height filtrations to reconstruct the original graph.

\paragraph{Our Contribution}
In this work, we focus on graphs embedded in $\R^d$ (plane graphs in $\R^2$)
and use directional
APDs to reconstruct a graph.
In particular, our main contributions are
an upper bound on the number of APDs required for reconstructing embedded graphs in $\R^d$,
a polynomial-time algorithm for reconstructing plane graphs,
and the first deterministic reconstruction algorithm for embedded graphs
in arbitrary dimension.

The current paper is an extension of conference proceedings from CCCG 2018~\cite{belton2018learning}.
We extend the proceedings paper in the following ways:
(1) we revise proofs for clarity;
(2) we extend our algorithms for graph reconstruction to~$\R^d$;
(3) we expand our literature review to include a discussion of recent results;
(4) we publicly release code for our
algorithm\footnote{The code is available in a git repo hosted on GitHub:
\url{https://github.com/compTAG/reconstruction}.}; and
(5) we provide an experimental section to demonstrate the implementation.

\section{Preliminaries}
\label{sec:preliminary}
 We begin by summarizing the necessary
background information, but refer the reader  to
\cite{edelsbrunner2010computational} for a more comprehensive overview of
computational topology.

\paragraph{Axis-Aligned Directions}
When considering a standard unit basis vector of~$\R^d$ in dimension $i \in
\{1,\ldots, d\}$, we
will write $e_i$ to denote the direction.
\paragraph{Plane Graphs}
Among our main objects of study are \emph{plane graphs} with straight-line embeddings
(referred to simply as \emph{plane graphs} throughout this paper). A plane graph
in $\R^2$ is
a set of vertices and a set of straight-line connections between pairs
of vertices called edges (denoted by $V$ and $E$ respectively)
such that no two edges in the embedding cross. We frequently
denote $|V|=n$ as the number of vertices in $\simComp$.
Throughout this paper, we make assumptions about the positioning
of vertices in graphs.

\begin{assumption}\label{assumpt:gp}
Let $\simComp$ be a graph embedded in
$\R^d$ with vertices $V$.
We assume that for all $x, y\in V$ where $x= (x_1, \ldots, x_d)$ and
\mbox{$y =(y_1, \ldots, y_d)$}, $x_i \neq y_i$ for any $i \in \{1, \ldots, d\}$.
Furthermore, we assume that no three vertices are collinear when projected
onto the subspace spanned by $e_1$ and $e_2$.
\end{assumption}

\paragraph{Height Filtration}
Let $\simComp$ be a plane graph. Consider a direction~$\dir$ in the unit sphere $\sph^{d-1}$ in
$\R^d$; we define the \emph{lower-star filtration}
with respect to direction $\dir$ in two steps. First, let
$\hFiltFun{\dir}: \simComp\rightarrow \R$ be defined for a simplex
$\sigma \subseteq \simComp$ by~\mbox{$\hFiltFun{\dir}(\sigma) = \max_{v \in \sigma}
\dprod{v}{\dir},$} where~$\dprod{v}{\dir}$ is the inner (dot) product and measures
height of $v$ in the direction of $\dir$, since~$\dir$ is a unit vector. Thus, the height
$\hFiltFun{\dir}(\sigma)$
 is the maximum height of all vertices in~$\sigma$.
Then, for each $t \in \R$, the
subcomplex~$\simComp_t:=\hFiltFunT{\dir}{(-\infty,t]}$ is composed of all
simplices that lie entirely below or at height~$t$, with respect to
direction $\dir$. Notice~$\simComp_r \subseteq \simComp_t$ for all~$r \leq t$
and~$\simComp_r=\simComp_t$ if no vertex has height in the interval $[r,t]$.
The sequence of all such subcomplexes, indexed by~$\R$, is the
\emph{height filtration} with respect to~$\dir$, denoted~$\hFilt{\dir} := \hFiltComplex{\dir}{\simComp}$.
Notice that the complex changes a finite number of times in this filtration.
We note that the lower-star filtration is the discrete analog to a
lower-levelset filtration, where the complex is intersected with a raising
closed half-plane.

\paragraph{Augmented Persistence Diagrams}
The persistence diagram for a filtered simplicial complex $K$
is a summary of the homology groups
as the parameter ranges from $-\infty$ to $\infty$; in particular,
the persistence diagram is a set of birth-death pairs of the
form~$(b,d)$ in the extended plane, each with a
corresponding dimension $k\in \Z_{\geq 0}$.  In particular, each pair
represents an independent generator of the $k$-th
homology group $H_k(K_t)$ for $t \in [b,d)$.
For technical reasons, all points in the diagonal~$y=x$ are also
included with infinite multiplicity.
In~\cite{belton2018learning}, we introduced the notion of the \emph{augmented
persistence diagrams} (APDs), which is persistence diagram with an additional
finite multi-set of points on the diagonal.  These points can be considered as
the points explicitly computed in an algorithm (such as the matrix reduction
described in~\cite{edelsbrunner2002topological}), but a formal definition using
the paired simplices of $K$ is provided in \cite[\S $2.5$]{mickaPhD}.
In the remainder of the paper,
when we use ``diagram,'' we mean APD; we write ``persistence diagram'' when we
mean the traditionally considered diagram. We denote the space of all APDs by $\mathcal{D}$.

For height filtrations of a graph embedded in $\R^d$, a
zero-dimensional birth
occurs when the height filtration discovers a new
vertex, representing a new connected component. A one-dimensional birth occurs
when a one-cycle is completed.
Zero-dimensional deaths correspond to connected components merging.
One-dimensional deaths are all at $\infty$.
All higher-dimensional homology groups are trivial.

For a direction $\dir \in \sph^{d-1}$, let the \emph{directional augmented
persistence diagram}~$\dgm{}{\hFiltComplex{\dir}{K}}$
 be the set of
birth-death pairs  from the height
filtration~$\hFiltComplex{\dir}{K}$.
Since each point in $\dgm{}{\hFiltComplex{\dir}{K}}$ has an associated dimension, we may restrict
our attention to the points in a specified dimension.  In particular, for $i \in
\Z$, let~$\dgm{i}{\hFiltComplex{\dir}{K}}$ be the diagram restricted to the
points with dimension $i$.\footnote{
    The diagram $\dgm{i}{\hFiltComplex{\dir}{K}}$ is actually a sub-diagram of
	$\dgm{}{\hFiltComplex{\dir}{K}}$.  As such, the set of all
    diagrams for a given direction is counted as one diagram, not one for each
    dimension.
}
As with the height filtration, we
simplify notation and define~$\dgm{i}{\dir} :=
\dgm{i}{\hFiltComplex{\dir}{K}}$ when the complex is clear from context.
We denote~$\beta_i$ to be the $i$-th Betti
number, i.e., the rank of the $i$-th homology group.
In general, the complexity of computing a diagram is matrix
multiplication time, with respect to the number $n$ of simplicies in the filtration;
that is, the complexity is $\Theta(n^{\omega})$, where $\omega$ corresponds
to the smallest known exponent for matrix multiplication time.
In some cases (e.g., for
computing $\dgm{0}{F}$ or $\dgm{d-1}{F}$ when $F$ is a height filtration
in~$\R^d$), the computation time is $\Theta(n\alpha(n))$, where $\alpha$ is the
inverse Ackermann function.

In what follows, for the unknown complex $K$,  we assume that we have an \emph{oracle}
$\mathcal{O}$ that takes a direction~$\dir \in \sph^{d-1}$ and produces the
diagram $\dgm{}{\dir}$
in that direction.  Additionally, we may restrict the
diagram to specific dimension(s) by specifying a dimension $i\in \Z$
and denoting the (sub-)diagram as $\dgm{i}{\dir}$.

We define $\persComp$ such that $\Theta(\persComp)$ is
the time complexity for $\mathcal{O}$ to return this~diagram.  Notice that
$\persComp = \Omega(| \dgm{}{\dir} |)$, where $|
\dgm{}{\dir} |$ denotes the number of off-diagonal points in~$\dgm{}{\dir}$.

Next, we state a lemma relating birth-death pairs.
We omit the proof, but refer the reader
to~\cite[pp.~$120$--$121$ of \S $V.4$]{edelsbrunner2010computational} for more
details.

\begin{lemma}[Adding a Simplex]
    Let  $k \in \Z_{\geq 0}$.
    Let $L \subset K$ be
    simplicial complexes that differ by a single
    $k$-simplex.
    Then, exactly one of the following is true:
    \begin{enumerate}
        \item $\beta_{k}(K) = \beta_{k}(L) + 1$,
        \item $\beta_{k-1}(K) = \beta_{k-1}(L)-1$.
    \end{enumerate}
\label{lem:addSimp}
\end{lemma}

As a result of \lemref{addSimp},
we can construct our filtration by adding one simplex at a time and
form a bijection between simplices of $K$ and birth-death events in the resulting
APD. If $G$ is any graph, then the maximum number of edges in $G$ is~$n(n-1)/2$,
and so \mbox{$|E|=O(n^2)$.}
In the
case when $G$ is a plane graph,~\mbox{$|E|=O(n)$} due to the planarity of $G$.
Furthermore, an APD will have at least~$n$ points from the
vertices in $G$ corresponding to births in the zero-dimensional diagram.
These
observations give us the following~corollary on the size of APDs for graphs.

\begin{corollary}[Size of Augmented Persistence Diagrams]
Let $G$ be an embedded graph and $n$ be the number of vertices in $G$. Then
an augmented persistence diagram has $O(n^2)$ birth-death pairs.
In the case when $G$ is a plane graph, the diagram
has $\Theta(n)$ birth-death pairs.
\label{rcor:PDpoints}
\end{corollary}

\section{Related Work}
In~\cite{turner2014persistent}
Turner et al., introduced the \emph{persistent homology transform (PHT)}
that represents a shape in $\R^d$---such as a simplicial complex---as a
family of persistence diagrams, parameterized by $\sph^{d-1}$.
In particular, the diagram for parameter $s \in \sph^{d-1}$ is the
persistence diagram defined by the height filtration in direction $s$.
The \emph{Euler characteristic transform (ECT)} is defined similarly
and maps a shape to a parameterized family of Euler characteristic curves~(ECCs),
where the ECC is the graph of Euler characteristic by filtration parameter, for
the same directional height filtrations.
Turner et~al.\ show that both of these maps are injective
for simplicial complexes in $\R^2$ or $\R^3$.
Recently, variations of the PHT and ECT have attracted interest in other
research domains and researchers are realizing the potential of persistent
homology as an effective data descriptor.
For example, Crawford et~al.~\cite{crawford2019predicting} introduces the
\emph{smooth Euler characteristic transform (SECT)} as a method of predicting
clinical outcomes using MRIs from patients with glioblastoma multiforme.
In~\cite{hofer2017deep}, persistence diagrams are used as features in deep
neural networks for classification of surfaces and graphs.
Additional injectivity and representation results for ECT and other topological
transforms are studied in~\cite{curry2018directions,ghrist2018persistent}.
Furthermore, a recent survey by Oudot and Solomon explores the current
state of inverse problems in topological persistence as a potential tool for
producing explainable data descriptors~\cite{outdot2018inverse}.

In order to leverage the injectivity of the PHT for shape comparison, two approaches can
be taken:
(1) show that the {PHT} has a finite representation;
(2) provide an algorithm that reconstructs the shape from a finite set of
directions.
In fact, (2) is a harder problem than (1), since an algorithm for reconstruction
will need to define a finite representation of the PHT.
One method to tackle approach (1) is to observe that diagrams only provide new information when a transposition
in the ordering of the filtration occurs; this observation is a direct result
of~\cite{cohen2007stability} and  guides the intuition
behind the current paper, as well
as~\cite{belton2018learning,curry2018directions,turner2014persistent}.
For example, consider a finite geometric simplicial complex in~$\R^d$ for
some $d \geq 2$.
Since transpositions can only happen when two vertices are transposed in the
filter, the set of directions in $\sph^{d-1}$ for which
two vertices occur at the same height is finite.
In other words, there are two hemispheres of~$\sph^{d-1}$
for two given vertices, say~$v$ and~$w$: one where $v$ is seen before $w$ and one where~$w$ is
seen before~$v$.
We take the hemispheres for all ${n}\choose{2}$ pairs of vertices
and consider
the regions for which the ordering is consistent.
For each such region, the persistence diagram continuously varies and no
transpositions (or `knees') are witnessed~\cite{cohen2006vines}.
In the current paper, we take approach (2) and show that we can use an oracle
to select a finite set of directions~$P \subset \sph^{d-1}$ that allow us to
reconstruct the original complex from the directional augmented persistence diagrams
from directions in~$P$.
In particular,  we
prove that a quadratic number of directions (with respect to the number of
vertices) is sufficient to reconstruct
a graph.

\section{Vertex Reconstruction}
\label{sec:vRec}

Next, we present an algorithm for recovering the locations of vertices of an
embedded graph. We begin with a plane graph~$\simComp$, where we are able to
use three directional augmented persistence diagrams. We then extend
this method for any embedded graph in $\R^d$, using $d+1$ directional augmented
persistence diagrams.


\subsection{Vertex Reconstruction for Plane Graphs}

The intuition behind vertex reconstruction is
that for each direction, we identify the lines on which the vertices of
$\simComp$  must lie. We show how to choose specific directions so that we can
identify all vertex locations by searching for points in the plane where three
lines intersect. We call these lines \emph{filtration~lines}:

\begin{definition}[Filtration Hyperplanes and Filtration Lines]
Given a direction~\mbox{$\dir \in \sph^{d-1}$} and a height $h \in~\R$, the
\emph{filtration hyperplane at height $h$} is the~$(d-1)$-dimensional
hyperplane, denoted $\pLine{\dir}{h}$,
through point~$h*\dir$ and perpendicular to direction $\dir$, where $*$ denotes
scalar multiplication. Given a finite
set of vertices~$V \subset~\R^d$, the \emph{filtration hyperplanes of~$V$}
are the set
of hyperplanes
$$
    \pLines{\dir}{V} := \{ \pLine{\dir}{\hFiltFun{\dir}(v)} \}_{v \in V}.
$$
In the special case when $d=2$, we refer to filtration hyperplanes
as \emph{filtration lines}.
\label{def:filtLines}
\end{definition}
By construction, all hyperplanes in $\pLines{\dir}{V}$ are parallel, and
$\pLine{\dir}{h}=\pLine{-\dir}{-h}$. In particular, if $v \in
V$, where~$V$ is the vertex set of a plane graph $\simComp$,
then the line $\pLine{\dir}{\hFiltFun{\dir}(v)}$ is perpendicular to $\dir$
and occurs at the height where the
filtration~$\hFilt{\dir}(\simComp)$ includes~$v$ for the first time.
In what follows, we adopt the following notation:
given a direction $s_i \in \sph^1$ and a point~$p \in \R^2$,
define~\mbox{$\simplePLine{i}{p} := \pLine{\dir_i}{\hFiltFun{\dir_i}(p)}$} as a way to
simplify notation.

By \lemref{addSimp}, the births in the zero-dimensional augmented persistence
diagram are in one-to-one correspondence with the vertices of the plane graph~$\simComp$.  
This means that, given a filtration line $\pLine{\dir}{h}$, exactly one vertex in $V$ lies 
on~$\pLine{\dir}{h}$. Using filtration lines from three directions, we show a one-to-one
correspondence between three-way intersections  of filtration lines and the
vertices in $\simComp$ in the next~lemma:

\begin{lemma}[Vertex Existence]\label{lem:vertExist}
    Let~$\simComp=(V,E)$ be a plane graph and let~\mbox{$n=|V|$.}  Let $\dir_1,\dir_2
    \in \sph^1$ be linearly independent and further suppose that
    $\pLines{\dir_1}{V}$ and~$\pLines{\dir_2}{V}$ each contain~$n$ lines. Let $A$
    be the set of $n^2$ intersection points between lines in~$\pLines{\dir_1}{V}$ and
    in~$\pLines{\dir_2}{V}$. Let $\dir_3 \in \sph^1$ such that each $a \in A$ has a
    unique height in direction $\dir_3$. Then, the following equality holds:
    $V = \pLines{\dir_3}{V} \cap A$.
\end{lemma}

\begin{proof}
    We prove this equality in two steps.
    First, we prove \mbox{$V \subseteq \pLines{\dir_3}{V} \cap A$.}
    Let~\mbox{$v \in V$.}  Then, for $i \in \{1,2,3\}$, there exists
    $\simplePLine{i}{v} \in \pLines{\dir_i}{V}$ such that~\mbox{$v \in
    \simplePLine{i}{v}$.}  Thus, $v \in \simplePLine{1}{v} \cap
    \simplePLine{2}{v} \cap \simplePLine{3}{v}$ and
    $$\simplePLine{1}{v} \cap \simplePLine{2}{v} \cap \simplePLine{3}{v}
    = \simplePLine{1}{v} \cap \bigg(\simplePLine{2}{v} \cap
    \simplePLine{3}{v}\bigg)
    \subset \pLines{\dir_i}{V} \cap A,$$
    by the definitions of filtration lines and $A$.
    Thus, $V  \subseteq \pLines{\dir_3}{V} \cap A$.

    Next, we prove $V \supseteq \pLines{\dir_3}{V} \cap A$.
    Since each $v \in V \subseteq A$ has a unique height in direction $s_3$ and
since every line of~$\pLines{\dir_3}{V}$ contains a vertex, we
    know that~$\pLines{\dir_3}{V}$ has $n$ lines.  Thus, we claim that each of
    these parallel lines intersects~$A$.
    Assume, for contradiction, that
    there exists~$\ell \in~\pLines{\dir_3}{V}$ such that $A \cap \ell = \emptyset$.
    Since $\ell$ is a filtration line in direction~$\dir_3$, there exists~\mbox{$v \in V$} such 
that~$\ell = \simplePLine{3}{v}$, meaning that this $v$ lies on $\ell$.
    However, $v \in  A$ since~\mbox{$V \subseteq \pLines{\dir_3}{V} \cap A \subseteq
    A$},
    contradicting the hypothesis that~$A\cap \ell = \emptyset$.
\end{proof}

If we generate vertical lines, $L_V = \pLines{e_1}{V}$, and
horizontal lines,~\mbox{$L_H = \pLines{e_2}{V}$,}
for our first two directions, then only
a finite number of directions in $\sph^1$ have been eliminated for the choice of
$s_3$.
In the next lemma, we choose a specific third direction by considering a
bounding region defined by
the largest distance between any two lines in $L_V$ and
smallest distance between any two consecutive lines in $L_H$.
Then, we pick the third direction so that if one of the corresponding
lines intersects the bottom left corner of this region then it will also intersect
the along the right edge of the region. In~\figref{vert}, the third direction
was computed using this procedure with the region having a width that is the length
between the left most and right most vertical lines, and height that is the length
between the top two horizontal lines. Next, we give a more precise description
of the vertex localization procedure.
\begin{figure}
\begin{center}
\includegraphics{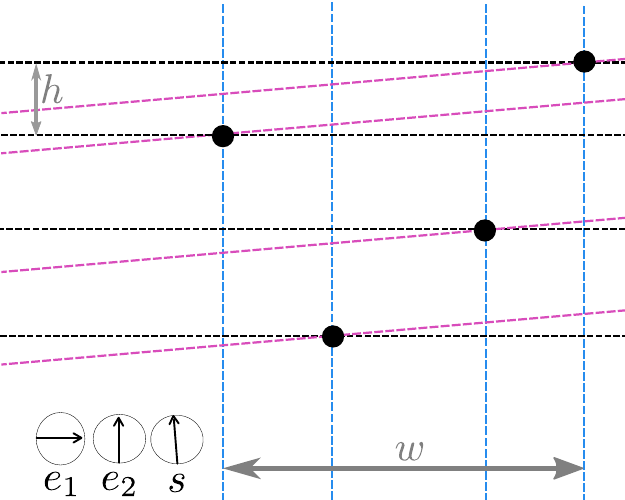}
    \caption{A vertex set $V$ of size four, with three sets of filtration lines.
    Here, notice that $e_1 \in \sph^1$ and $e_2 \in \sph^1$ are linearly
independent, and the third direction~$\dir$
satisfies the assumptions of \lemref{vertExist}.
The lines of~$\pLines{e_1}{V}$ are
the blue vertical lines,~$\pLines{e_2}{V}$ are the black horizontal lines,
    and~$\pLines{\dir}{V}$ are the pink diagonal lines. The three-way intersection
points (one from each set of filtration lines) is in one-to-one
correspondence to the vertices in $V$.
Lines~\ref{line:vert:2sort}
through~\ref{line:vert:lastdir} in Algorithm~\ref{alg:vert_recon} provide details
for finding $s_3$ using the width of the
vertical lines (marked~$w$) and the minimum height difference between horizontal
    lines (marked $h$).}
\label{fig:vert}
\end{center}
\end{figure}

\begin{lemma}[Vertex Localization]
    Let $L_H$ and $L_V$ be~$n$ horizontal and $n$ vertical lines, respectively.
    Let~$w$ (and $h$) be the largest (and smallest) distance between two lines of
    $L_V$ (and $L_H$, respectively). Let $B$ be the smallest axis-aligned
    bounding region containing the intersections of lines in $L_H \cup L_V$.
    Let \mbox{$\dir = (w, h/2)/||(w,h/2)||$,} i.e., a unit vector oriented towards the point $(w, h/2)$.
    Any line parallel to $\dir$ can intersect at most one line of $L_H$ in $B$.

    \label{lem:vertex-localization}
\end{lemma}

\begin{proof}
    Note that, by definition, $\dir$ is a vector in the direction that is at a slightly
    smaller angle than the diagonal of the region with width $w$
    and height~$h$.
    Assume, by contradiction, that a line parallel to $\dir$ can intersect two
    lines of $L_H$ within~$B$. Specifically, let $\ell_1, \ell_2 \in L_H$ and
    let~$\ell_{\dir}$ be a line parallel to $\dir$ such that the points
    $\ell_i \cap \ell_{\dir} = (x_i, y_i)$ for~$i = \{1, 2\}$ are the two such intersection
    points within~$B$. Since the lines of $L_H$ are horizontal
    and by the definition of~$h$, we observe that $|y_1-y_2|\geq h$.
    Let $w' = |x_1 - x_2|$, and observe $w' \leq w$.  Since the slope of $\ell_s$ is
    $\frac{h}{2w}$, 
    we have~$|y_1 - y_2| < h$, which is a contradiction.
\end{proof}

We conclude the discussion of plane graph reconstruction with an algorithm to
compute the coordinates of the vertices
of the original graph in~$\R^2$, using only three diagrams.

\begin{theorem}[Vertex Reconstruction]
    Let~$\simComp$ be a plane graph. We can compute the coordinates of all
    $n$ vertices of~$\simComp$ using three directional augmented persistence diagrams
    in $\Theta(n \log{n} + \persComp)$ time, where $\Theta(\persComp)$ is the time complexity of
    computing a single directional augmented
    persistence~diagram for $G$.
\label{thm:intComp}
\end{theorem}

\begin{proof}
    We proceed with a constructive proof that is presented as an
    algorithm in \algref{vert_recon}. Let
    $\mathcal{O}$ be an oracle that takes a
    direction~$\dir \in \sph^{1}$ and returns the
    zero-dimensional directional APD for the unknown plane graph~$\simComp$
    in direction~$\dir$ in $\Theta(\persComp)$ time.

    We start with requesting two directional augmented persistence diagrams from
    the oracle, $\dgm{0}{e_1}$ and $\dgm{0}{e_2}$.
    Note that, by the \assumptref{gp}, no two
    vertices of~$\simComp$ share an $x$- or $y$-coordinate.
    By \rcorref{PDpoints}, the
    sets~$\pLines{e_1}{V}$ and~$\pLines{e_2}{V}$ (which we do not explicitly
    construct) each contain~$n$
    distinct lines.
    Let $L_1=\{ h_1, h_2, \ldots, h_n \}$
    be the resulting set of heights of lines in~$\pLines{e_1}{V}$,
    in increasing order. Likewise, let~$L_2=\{ h'_1,
    h'_2, \ldots, h'_n \}$ be the ordered set of heights of lines
    in~$\pLines{e_1}{V}$, also in increasing order.
    As a consequence of \lemref{addSimp}, the birth times in the zero-dimensional
    diagrams are in one-to-one correspondence to heights of the
    filtration lines.
    Thus, we can compute $L_1$ and $L_2$ from the APD in direction
    $\dir$ in $\Theta(n)$~time by iterating through the points in the
    zero-dimensional augmented persistence diagram, then sorting in
    $\Theta(n\log n)$ time.

    Let $A$ be the set of $n^2$ intersection points between the lines in
    $\pLines{e_1}{V}$ and in~$\pLines{e_2}{V}$.  Exactly $n$ of these points
correspond to vertices of $V$.
    The next step is to identify
    a third direction $\dir$ such that each line in~$\pLines{\dir}{V}$ intersects
    with only one point in $A$, which we will use in order to distinguish which
    $n$ intersection points correspond to vertices in $V$.

    Let $w=h_n-h_1$ and let~$h$ be the minimum of $\{ h'_i -h'_{i-1} \}_{i=2}^n$.
    In words, $w$ is the difference between the maximum and minimum heights of
    lines in
    $\pLines{e_1}{V}$  and $h$ is the minimum
    height difference between consecutive lines in~$\pLines{e_2}{V}$; see \figref{vert}.
    Note that we can compute $w$ in~$\Theta(1)$ time
    from $L_1$ and $h$ in $\Theta(n)$ time from $L_2$.
    Let $B$
    be the smallest axis-aligned bounding region containing the intersection
    points~$A$, and let
    $\dir$ be a unit vector perpendicular to the vector~$[w, \frac{h}{2}]$.
    We request the set $\dgm{0}{s}$ from our oracle $\mathcal{O}$.
    As before, the heights of the lines in $\pLines{s}{V}$ are the birth times of
    points in $\dgm{0}{s}$.  We save this set of heights as $L_s$ in $\Theta(n)$
    time, and sort $L_s$ in $\Theta(n \log n)$ time.

    Finally, by~\lemref{vertex-localization},
    any line in~$\pLines{s}{V}$ intersects no more than one line
    of~$\pLines{e_2}{V}$ within $B$. Furthermore, by~\lemref{vertExist},
    the $n$ vertices in $G$ are in one-to-one correspondence with points in
    $\pLines{\dir}{V} \cap A$.
    We compute these three-way intersections of $\pLines{\dir}{V} \cap A$
    by intersecting the $i$-th line
    of~$\pLines{e_2}{V}$ with the $i$-th line of $\pLines{\dir}{V}$ in~$\Theta(n)$~time.

    In total, this algorithm, summarized in \algref{vert_recon},
    uses three directional diagrams, two requested from
    the oracle in Line \ref{line:vert:2oracle} and one
    requested in Line~\ref{line:vert:1oracle}.
    These two lines take
    $\Theta(\persComp)$ time each,
    Lines \ref{line:vert:2sort} and \ref{line:vert:1sort} take $\Theta(n
    \log n)$ time each, and the for loop in Lines
    \ref{line:vert:forstart} through \ref{line:vert:forend}
    takes $\Theta(n)$ time.  All other lines are linear or constant,
    with respect to $n$.
    Thus, the total time complexity
    is~$\Theta(n \log n  + \persComp)$.
\end{proof}
\begin{algorithm}
\caption{Reconstruct Vertices}
\label{alg:vert_recon}
\begin{algorithmic}[1]
    \REQUIRE Oracle $\mathcal{O}$ for an unknown graph $G=(V,E) \subset \R^2$.
    \ENSURE Set of vertex locations.
    \STATE Consult $\mathcal{O}$ to obtain diagrams $\dgm{0}{e_1}$ and $\dgm{0}{e_2}$
        \label{line:vert:2oracle}
    \STATE $L_1 \gets$ birth times from points in $\dgm{0}{e_1}$
    \STATE $L_2 \gets$ birth times from points in $\dgm{0}{e_2}$
    \STATE Sort $L_1$ and $L_2$ in increasing order\label{line:vert:2sort}
    \STATE $w \gets$ maximum minus minimum in $L_1$
    \STATE $h \gets$ minimum gap between two consecutive values in $L_2$
    \STATE $\dir \gets$ a unit vector perpendicular to the vector $[ w, h/2 ]$\label{line:vert:lastdir}
    \STATE Consult $\mathcal{O}$ to obtain diagram $\dgm{0}{s}$ \label{line:vert:1oracle}
    \STATE $L_s \gets$ birth times from points in $\dgm{0}{s}$
    \STATE Sort $L_s$ in increasing order\label{line:vert:1sort}
    \FOR{$i=1, 2, \ldots, n$}\label{line:vert:forstart}
    \STATE $\ell_i \gets$ horizontal line with the $i$-th element of
    $L_2$ as the $y$-coordinate
    \STATE $\ell_i' \gets$ line perpendicular to $s$ at height
    equal to the $i$-th element of $L_s$
    \STATE $v_i = \ell_i \cap \ell_i'$
    \ENDFOR \label{line:vert:forend}
    \RETURN $\{ v_i \}_{i=1}^n$
\end{algorithmic}
\end{algorithm}

\subsection{Vertex Reconstruction in $\R^d$}

The vertex reconstruction algorithm of the previous subsection generalizes to higher
dimensions.  In $\R^d$, a filtration line becomes a \emph{filtration hyperplane},
a~\mbox{$(d-1)$-dimensional} hyperplane that goes through one of the vertices in the
vertex set (and is perpendicular to a given direction). Similar to filtration lines, filtration hyperplanes
generated by a fixed direction are parallel and are in a one-to-one
correspondence with the vertices (for almost all directions).

\begin{lemma}[Generalized Vertex Existence]\label{part:genallIntersect}
    Let~$\simComp=(V,E)$ be a straight-line embedded graph in $\R^d$.
    Let $\dir_1,\dir_2,\ldots, \dir_d$ be linearly independent directions in~$\sph^{d-1}$
    and further suppose
    that $\pLines{\dir_i}{V}$
    contains~$n$ filtration hyperplanes for each \mbox{$i \in \{1, 2, \ldots,d\}$.}
    Choosing one hyperplane in each set~$\pLines{s_i}{V}$,
    the intersection of these hyperplanes is a point.
    Let~$A$ denote the $n^d$ such intersection points.
    Let $\dir_{d+1} \in \sph^{d-1}$ such that each $a \in A$ has a unique height in
    direction~$\dir_{d+1}$.  Then, the following equality holds:
    $V = \pLines{\dir_{d+1}}{V} \cap A$.
\label{lem:genVertExist}
\end{lemma}
\begin{proof}
    This proof follows the same structure of the proof of \lemref{vertExist},
    generalizing $\R^2$ to $\R^d$.

    We prove this equality in two steps.
    First, we prove $V \subseteq \pLines{\dir_{d+1}}{V} \cap A$.
    Let $v \in V$.  Then, for $i \in \{1,2,\ldots, d + 1\}$, there exists
    $\simplePLine{i}{v} \in \pLines{\dir_i}{V}$ such that~\mbox{$v \in
    \simplePLine{i}{v}$.}  Thus, $v \in \cap_{i=1}^{d+1} \simplePLine{i}{v}$, as was to be shown.

    Next, we prove $V \supseteq \pLines{\dir_{d+1}}{V} \cap A$.
    Assume, for contradiction, that
    there exists a filtration hyperplane~$\ell \in~\pLines{\dir_{d+1}}{V}$ such
    that $A \cap \ell = \emptyset$.
    Since~\mbox{$\ell \in \pLines{\dir_{d+1}}{V}$}, there exists a vertex
    $v \in V$ such that~$\ell = \simplePLine{d+1}{v}$ and~$v$ lies on $\ell$.
    However, since $V \subset A$ from above,
    $v$ is in $\cap_{i=1}^d \simplePLine{i}{v} = A$, contradicting the
    hypothesis $A \cap \ell = \emptyset$.
\end{proof}

Just as in the case of plane graphs, we can now describe a method for locating
all vertices. The following lemma is a higher-dimensional analogue
of~\lemref{vertex-localization}.

\begin{lemma}[Generalized Vertex Localization]
    Let $\simComp = (V,E)$ be a straight-line embedded graph in $\R^d$,
    with $n=|V|$.
    Choosing one hyperplane in each set~$\pLines{e_i}{V}$ for $1\le i \le d$,
    the intersection of these hyperplanes is a point.
    Let~$A$ denote the $n^d$ such intersection points.
    Then, we can find a direction $s\in \sph^{d-1}$ such that each
    hyperplane in $\pLines{s}{V}$
    intersects at most one of the $n^d$ points in~$A$ in~$\Theta(dn\log n)$ time.
    \label{lem:genVertLoc}
\end{lemma}
\begin{proof}
    Let~$L_i=\{\hVarForDirI{1},
    \hVarForDirI{2}, \ldots, \hVarForDirI{n} \}$ be the ordered set of heights of hyperplanes
    in~$\pLines{e_i}{V}$.
    Let $\wForDirI = \hVarForDirI{n} - \hVarForDirI{1}$; in other words, $\wForDirI$ is the largest distance between
    any two hyperplanes in the set $\pLines{e_i}{V}$,
    and let $w = \max_{1\le i\le d} \wForDirI$.
    Let~$\hForDirI = \min \{ \hVarForDirI{j} - \hVarForDirI{j-1} \}_{j=2}^n$;
    in other words, $\hForDirI$ is
    the smallest height difference between any two (adjacent)
    hyperplanes in the set $\pLines{e_i}{V}$, and let~$h= \frac{1}{2}\min_{1\le i\le d} \hForDirI$.
    Then, we consider the hyperplane that intersects the origin and each
    point~$we_i + \langle 0,\ldots 0,\frac{h}{d-1} \rangle$ for $i\in \{1,2,\ldots,
    d-1\}$.
    Next, we choose a vector orthogonal to the hyperplane.
    Without loss of generality, we choose
     \[
     x = \Bigl\langle -\frac{1}{w}, \ldots,
    -\frac{1}{w}, \frac{d-1}{h} \Bigr\rangle^{\top}
    \]
    and observe that
     $(we_i + \langle 0,\ldots, 0,\frac{h}{d-1} \rangle) \cdot x = 0$
     for each $i \in \{1,2,\ldots, d-1\}$. Finally, we define~$\dir =
     \frac{x}{||x||}$.

We now show that $\dir$ satisfies the claim that each hyperplane in $\pLines{s}{V}$
    intersects at most one of the $n^d$ points in~$A$ and that we can find
	$\dir$ in
    $\Theta(dn\log n)$ time.
    Let $p=(p_1,p_2, \ldots, p_d)$ and $q=(q_1,q_2, \ldots, q_d)$ be points
    in $A$ with $p \neq q$, and assume, for contradiction,
that they lie on the same
hyperplane in~$\pLines{\dir}{V}$.
Then, $p \cdot \dir = q \cdot \dir$.  By the definition of dot product, we have
    the following equation (after multiplying $\dir$ by $||x||$):
\begin{align*}
    \frac{d-1}{h}(p_d - q_d) - \sum_{i=1}^{d-1} \frac{1}{w} (p_i - q_i) &=0.
\end{align*}
    Since $\frac{d-1}{h}$ and $w$ are positive numbers, we can rearrange this
    equality to obtain:
    \begin{align}
        \left|p_d-q_d\right|
        &=
        \frac{h}{w(d-1)} \left|\sum_{i=1}^{d-1}(p_i - q_i)\right|
        .\label{eq:equality}
\end{align}

    Recall that $h= \frac{1}{2}\min_{1\le i\le d} h_i \leq \frac{1}{2}h_d$. Therefore, we know
    that~\mbox{$2h \leq h_d \leq |p_d - q_d|$.}  Applying \eqnref{equality} to this inequality, we obtain:
\begin{align*}
    2h &\le \frac{h}{w(d-1)} \left|\sum_{i=1}^{d-1}(p_i - q_i)\right|\\
       &\le \frac{h}{w(d-1)} (d-1)\max_{1 \le i \le d-1} |(p_i - q_i)|\\
       &\le \frac{h}{w} w.
\end{align*}
Thus, we have $2h < h$,
which is a contradiction when $n\geq 2$.  For the case where~$n=1$, the vertex
    is $(\hOneForDir{1}, \hOneForDir{2}, \ldots, \hOneForDir{d})$ and the $(d+1)^{\text{st}}$ direction is not needed.

We analyze the complexity of computing $w$ and $h$.
For each direction $e_i$, we perform three steps.
First, we sort the heights of $\pLines{e_i}{V}$ in $\Theta(n \log n)$ time.
Second, we compute $\wForDirI$ in constant time
(as it is the maximum value minus the minimum value of the heights).
Third, we compute $\hForDirI$ in $\Theta(n)$ time.
As there are $d$ dimensions, computing the sets $\{ \wForDirI \}$ and
$\{ \hForDirI \}$ takes~$\Theta(d n\log n)$ time.
Computing $w$ and $h$ from the sets $\{ \wForDirI \}$ and $\{ \hForDirI \}$ is $\Theta(d)$
time.
Thus, the bottleneck is sorting in each direction, which makes the total runtime~$\Theta(dn\log n)$.
\end{proof}

Equipped with the above method for finding a suitable
$(d+1)$\textsuperscript{st} direction to locate vertices in higher dimensions,
we conclude this section with a theorem describing the algorithm to
compute the coordinates of the vertices of the original embedded graph.

\begin{theorem}[Vertex Reconstruction in Higher Dimensions]
    Let~$\simComp$ be a straight-line embedded graph in~$\R^d$ for $d>1$.
    We can can compute the coordinates of all
    $n$ vertices of $\simComp$ using
    $d+1$ directional augmented persistence~diagrams
    in~$\Theta(dn^{d+1} + d\persComp)$ time,
    where $\Theta(\persComp)$ is the time complexity of computing
    a persistence diagram.
\label{thm:genIntComp}
\end{theorem}
\begin{proof}
    We proceed with a constructive proof, generalizing the constructive proof
    from~\thmref{intComp} and \algref{vert_recon}. Let
    $\mathcal{O}$ be an oracle that takes a direction~\mbox{$\dir \in \sph^{d-1}$} and returns the
    $\dgm{0}{s}$ in $\Theta(\persComp)$ time.

    For~$i \in \{1, \dots, d\}$, we use this oracle to obtain $\dgm{0}{e_i}$.
    Note that, by the \assumptref{gp} and \rcorref{PDpoints},
    for each of these directions, we have
    exactly $n$ distinct filtration hyperplanes,
    in one-to-one correspondence with the vertices.  Note that, for a given
    direction $e_i$, we
    store the filtration hyperplanes as a list of the vertex heights.
    Choosing one hyperplane in each direction yields~$d$ pairwise orthogonal
    hyperplanes; their intersection is a point in~$\R^d$ and this point is a
    potential vertex location.  In total, we
    have~$n^d$ potential vertex locations, of which only~$n$ are actual
    vertices. We
    denote this set of~$n^d$ potential vertex locations by~$A$.
    Then, computing these lists of vertex heights takes~$\Theta(\persComp)$ 
    time per dimension to account for computing and listing
    the points of the APD.

    Let $\dir$ be chosen as in \lemref{genVertLoc} in~$\Theta(dn\log n)$ time.
    By \lemref{genVertLoc}, each hyperplane~$\simplePLine{\dir}{v}$ intersects
    at most one point in~$A$ for each~$v\in V$.  Thus, by
    \lemref{genVertExist}, there are exactly~$n$ distinct
    intersections between~$\pLines{\dir}{V}$ and $A$, in one-to-one correspondence
    with the~$n$ vertices of $\simComp$.

    Then, to identify vertex locations in~$\R^d$, we employ the following
    brute force algorithm. We
    check each element~$v\in A$ for intersections with any
    hyperplane~$\ell \in \pLines{\dir}{V}$. Since~$|A| = n^d$
    and~$|\pLines{\dir}{V}| = n$,
    we have~$n^{d+1}$ checks that we must perform, with each
    check taking $\Theta(d)$ time. Thus, the total time complexity
    of calculating $V$ from the~$d+1$ sets of filtration hyperplanes
    is~$\Theta(dn^{d+1})$ and no additional augmented
    persistence diagrams are computed.

    In total, this algorithm uses $d+1$ directional diagrams. The time
    complexity of constructing the~$d+1$ sets of filtration hyperplanes
    is~$\Theta(dn \log n + d\persComp)$, and an additional~$\Theta(dn^{d+1})$ time
    to compute the actual vertex locations.
    Thus, the total time complexity is~$\Theta(dn^{d+1} + d\persComp)$.
\end{proof}

\section{Edge Reconstruction}
\label{sec:eRec}

Given the vertices constructed in \secref{vRec}, we describe how to reconstruct
the edges in an embedded graph using~$n^2-n$ augmented persistence diagrams.
The key to determining whether an edge exists or not is counting the degree of
a vertex for edges in the half plane ``below'' the vertex with respect to a
given direction. We begin with a method for reconstructing plane graphs,
and then extend our method to embedded graphs in $\R^d$.

\subsection{Edge Reconstruction for Plane Graphs}

We first define necessary terms, and then describe our algorithm for
constructing edges.

\begin{definition}[Indegree of Vertex]
Let $\simComp$ be a straight-line embedded graph in $\R^d$
with vertex set~$V$. Then, for every vertex
$v \in V$ and every direction~\mbox{$\dir \in \sph^{d-1}$,} we define:
$$\indeg{v}{\dir} = |\{(v,v') \in E \mid \dir \cdot v' \leq \dir \cdot v \}|.$$
\end{definition}
Thus, the indegree of $v$ is the number of edges incident to $v$ that lie
below~$v$, with respect to direction~$\dir$; see~\figref{indegree}.
\begin{figure}[htb]
\begin{center}
\includegraphics{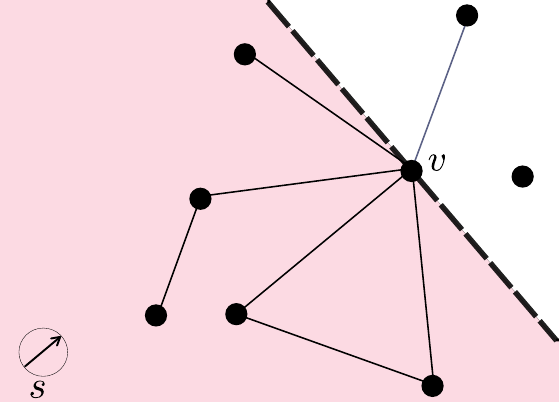}
\caption{A plane graph with a dashed line drawn intersecting $v$ in the direction
perpendicular to~$\dir$. Since four edges incident to $v$  lie below $v$ indicated
by the shaded region, with respect to direction~$\dir$, $\indeg{v}{\dir}=4$.
    }
\label{fig:indegree}
\end{center}
\end{figure}

Given a directional augmented persistence diagram, we prove that we can compute
the indegree of a
vertex with respect to that direction:

\begin{lemma}[Indegree from Diagram]\label{lem:Indegree}
    Let $\simComp=(V,E)$ be a straight-line embedded graph in $\R^d$.
    Let~$\dir \in~\sph ^{d-1}$
    such that no two vertices have the same height with respect to $\dir$ (and
    thus
    $|\pLines{\dir}{V}| = n$). Let~$\dgm{0}{\dir}$ and~$\dgm{1}{\dir}$ be the
    zero- and one-dimensional points of the
    augmented persistence diagram resulting from the height filtration
    $\hFilt{\dir}(\simComp)$. Then, for all $v \in V$,
\begin{equation*}
\begin{split}
\indeg{v}{\dir} = & |\{(x,y) \in \dgm{0}{\dir} \mid y = v \cdot \dir \}|
+|\{(x,y) \in \dgm{1}{\dir} \mid x = v \cdot \dir \}|.
\end{split}
\end{equation*}
Furthermore, if $n=|V|$ and $d=2$ then $\indeg{v}{\dir}$ can be
computed in $\Theta(n)$ time. If $d>2$, then $\indeg{v}{\dir}$ can be
computed in $O(n^2)$ time.
\end{lemma}

\begin{proof}
Let $v, v' \in V$ such that $\dir \cdot v' < \dir \cdot v$, i.e., the vertex~$v'$
is lower than $v$ in direction~$\dir$. Let $e =(v, v') \in E$. Then, by~\lemref{addSimp},
    we have two cases to consider when $e$ is added to $F_s$:

    \emph{Case 1: $e$ joins two disconnected components.}  If $e$ connects two
    previously disconnected components, then $e$ is associated with a death in
    $\dgm{0}{\dir}$ at height~\mbox{$\dir \cdot v$.}  Moreover, since all deaths in
    $\dgm{0}{\dir}$ are associated with adding an edge, we know that the set of
    all edges that fall into this case with $v$ as the top endpoint is $A= \{(x,y)
    \in \dgm{0}{\dir} \mid y = v \cdot \dir \}$.

    \emph{Case 2: $e$ creates a one-cycle.}  In this case, $e$ is associated
    with a birth in~$\dgm{0}{\dir}$ at height~$\dir \cdot v$.  Thus, we have
    that~$B=\{(x,y) \in \dgm{1}{\dir} \mid x = v \cdot \dir \}$ is the set of
    edges that fall into
    this case with $v$ as the top endpoint.

The union $A \cup B$ is the set of all edges ending at~$v$ with respect to~$s$,
hence~$\indeg{v}{\dir} = |A \cup B|$.
Furthermore, by~\rcorref{PDpoints}, if $d=2$, then $G$ is a plane graph and
each of $\dgm{0}{\dir}$ and $\dgm{1}{\dir}$ have
$\Theta(n)$ points, so we count and sum the points joining two disconnected
components or creating one-cycles at height $v$ in time $\Theta(n)$ time. If
$d>2$, then each of $\dgm{0}{\dir}$ and $\dgm{1}{\dir}$ have
$O(n^2)$ points by~\rcorref{PDpoints},
so we count and sum the points joining two disconnected
components or creating one-cycles at height $v$ in time $O(n^2)$ time.
\end{proof}

In order to decide whether an edge $(v,v')$ exists between two vertices, we look at the
degree of $v$ as seen by two close directions such that $v'$ is the only vertex in what
we call a \emph{wedge at $v$}:
\begin{definition}[Wedge]
    Let ~$v \in V$, and choose $\dir_1,\dir_2 \in \sph^{d-1}$.  Then, a
    \emph{wedge} at $v$ is the closure of the
    symmetric difference between the half planes below $v$ in directions
    $\dir_1$ and~$\dir_2$. In the special case when $d=2$, we refer
    to the wedge as a \emph{bow tie}.
\end{definition}

In what follows, we use a wedge that is defined by a direction perpendicular to
an edge and slight tilts of that direction.
In the proof of \cite[Theorem 3.1]{turner2014persistent},
Turner et al.\ also develop a construction
that is conceptually similar to the wedge for identifying links
 where changes in Betti
numbers are observed when considering the effect of the continuous rotation of directions above
and below the plane orthogonal to an edge.

Because we assume that no three vertices in our plane graph are collinear, for each
pair of vertices~$v,v'\in V$, we always find a bow tie centered at $v$ that
contains the vertex~$v'$ and no other vertex in $V$; see
\figref{edgeExist}.
We use bow tie regions to determine if there exists an edge between $v$ and $v'$.
In the next lemma, we show how to decide if the edge~$(v,v')$ exists in our plane graph.
\begin{figure}[htb]
\begin{center}
\includegraphics{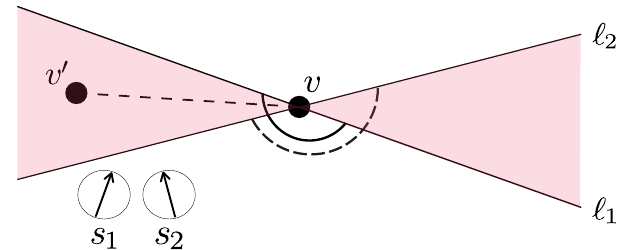}
\caption{ Bow tie~$\mathbb{B}$ at~$v$, denoted by the shaded area.~$\mathbb{B}$ contains
exactly one vertex,~$v'$, so the only potential edge in~$\mathbb{B}$ is $(v,v')$.
In order to determine if there exists an edge between~$v$ and~$v'$, we
compute~$\indeg{v}{\dir_1}$ and~$\indeg{v}{\dir_2}$, i.e., the number of edges
incident to~$v$ in the solid and dashed arcs, respectively. An edge
exists between~$v$ and~$v'$ if and only if~$|\indeg{v}{\dir_1} -
    \indeg{v}{\dir_2}|=1$.}
\label{fig:edgeExist}
\end{center}
\end{figure}

\begin{lemma}[Edge Existence] \label{lem:edgeExist}
    Let $\simComp=(V,E)$ be a straight-line embedded graph in $\R^d$.
Let~\mbox{$v,v' \in V$}. Let $\dir_1, \dir_2 \in  \sph ^{d-1}$ such that the
wedge~$\mathbb{B}$ at $v$ defined by $\dir_1$ and $\dir_2$ satisfies:~$\mathbb{B}
\cap V\setminus\{v\} = v'$.~Then,
$$\big|\indeg{v}{\dir_1} - \indeg{v}{\dir_2}\big|=1 \iff (v,v') \in E.$$
\end{lemma}
\begin{proof}
    Since edges
    in~$\simComp$ are straight lines, any edge incident to~$v$ either falls
    entirely in the wedge region~$\mathbb{B}$ or the interior of the
    edge is on the same
    side (above or below) of
    both hyperplanes that define~$\mathbb{B}$.
    Let $A$ be the set of edges that are incident to $v$ and below both
    hyperplanes; that is, $$A = \{ (v,w) \in E ~|~
        \dprod{\dir_1}{w} < \dprod{\dir_1}{v} \text{ and }
        \dprod{\dir_2}{w} < \dprod{\dir_2}{v}\}.$$
    Next, we split the wedge into the
    two infinite cones defined by the two connected components of the interior
    of $\mathbb{B}$. Let $\mathbb{B}_1$ be the set of
    edges in one cone and~$\mathbb{B}_2$
    be the set of edges in the other cone.
    Then, by definition of indegree,
    \begin{align*}
        \big|\indeg{v}{\dir_1} - \indeg{v}{\dir_2}\big|
        &= \big||A|+|\mathbb{B}_1|-|A|-|\mathbb{B}_2| \big|\\
        &= \big| |\mathbb{B}_1|-|\mathbb{B}_2| \big|.
    \end{align*}
    As~$\mathbb{B}$ contains one vertex~$v'$,
    $\left| |\mathbb{B}_1| - |\mathbb{B}_2| \right|$
    is one when $(v,v')\in E$, and zero otherwise.
    Therefore, we conclude that~$|\indeg{v}{\dir_1} - \indeg{v}{\dir_2}| = 1 \iff (v,v') \in E$,
    as required.
\end{proof}

Next, we prove that we can find the embedding of the edges in plane graphs using $\Theta(n^2)$ directional augmented persistence diagrams.
See \appendref{Example} for an example of
walking through the reconstruction.

\begin{theorem}[Edge Reconstruction]
\label{thm:edgeEmbed}
    Let $G=(V,E)$ be a plane graph.
    If~$V$ is known, then we can compute~$E$ using~$n^2-n$ directional
    augmented
persistence~diagrams in $\Theta(n^2\persComp)$ time,
    where~$\Theta(\persComp)$ is the time
    complexity of computing a single~diagram.
\end{theorem}

\begin{proof}
    We prove this theorem constructively, and summarize the construction in
    \algref{edge_recon}. In the algorithm, we first preprocess in order to find a
    global bow tie half-angle. Then we iterate through each pair of vertices
    and test to see if the edge exists.  This edge test is done in two steps:
    first create a bow tie that isolates the potential edge, then apply \lemref{edgeExist} to determine if the edge
    exists or not by comparing the indegrees of $v$ with respect to the two
    directions defining the bow tie.

    \emph{Preprocessing
    (Lines~\ref{line:edge:firstline}--\ref{line:edge:3minang} of \algref{edge_recon}).}
    We initialize a set $E$ of edges to be the empty set in
    Line~\ref{line:edge:firstline}.  Next, we compute an angle that is
    sufficiently small to be used to construct bow ties for every edge.
    For each vertex~$v\in V$, we consider the cyclic ordering of the points
    in~$V \setminus \{v\}$ around $v$; let $c[v]$ denote this ordered list of
    vertices.
    By Lemmas~$1$ and~$2$ of~\cite{verma2011slow}, we compute~$c[v]$ for all~$v \in V$
    in~$\Theta(n^2)$ total time\footnote{Note that the na\"ive approach would
    be to sort about each vertex
    independently, which would take $\Theta(n^2 \log n)$ time, but the results
    of~\cite{verma2011slow} improve this to $\Theta(n^2)$.
    The lemmas in \cite{verma2011slow} use
    big-O notation, but the presented algorithm is actually asymptotically
    tight.}.
    Once we have these cyclic orderings, we  compute all~$n$ lines through
    $v$ and~$v_i \in V \setminus \{v\}$ and can compute a cyclic ordering of all
    such lines through $v$ in $\Theta(n)$ per vertex.  (The step of obtaining the
    cyclic ordering of lines given the cyclic ordering of vertices is similar to
    the merge step of merge sort).
    Given two adjacent lines through $v$, consider the angle between these
    lines; see the angles labeled
    $\theta_i$ in \figref{edgebow}. For each vertex $v$, the minimum of such angles, denoted
    $\theta(v)$, is computed in Line~\ref{line:edge:2minang} in $\Theta(n)$
    time.
    Finally, we define $\theta = \frac{1}{2} \min _{v \in V} \theta(v)$ in
    Line~\ref{line:edge:3minang}.  The value $\theta$ will be used to compute
    bow ties in the edge test.
    The
    runtime for this preprocessing is~$\Theta(n^2)$ and
    requires no augmented persistence~diagrams.
    \begin{figure}
        \center
        \includegraphics{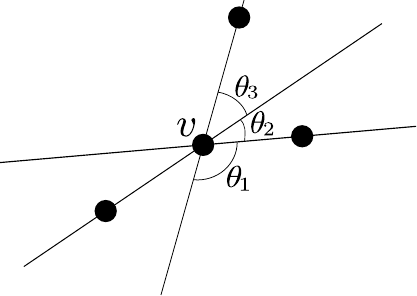}
        \caption{Ordering of all vertices about $v$.
        Lines are drawn through all vertices and then angles are computed
        between all adjacent pairs of lines. The smallest angle is denoted
        as~$\theta(v)$. Here,~$\theta(v)=\theta_2$.}
        \label{fig:edgebow}
    \end{figure}

    \emph{Edge Test
    (Lines~\ref{line:edge:edgeTestStart}--\ref{line:edge:edgeTestEnd})}.
    Let
    $\mathcal{O}$ be an oracle that takes a direction $\dir$ in $\sph^{1}$ and returns the APD for the
    unknown plane graph~$\simComp$
    in direction $\dir$ in time~$\Theta(\persComp)$.
    Let $v,v' \in V$ such that~$v \neq v'$.  We now provide the two steps
    necessary to test if $(v,v') \in E$ using only two diagrams.

    The first step is to construct bow ties
    (Lines~\ref{line:edge:bowtiestart}--\ref{line:edge:bowtieend} of \algref{edge_recon}).
    Let~$\dir$ be a unit vector perpendicular to \mbox{vector~$v'-v$},
    and~let $\dir_1,\dir_2$ be the two unit vectors that form
    angles~$\pm {\theta}$ with~$\dir$. Note that~$v$ and~$v'$ are at the same
    height in direction~$\dir$, but different heights in direction $\dir_1$ and
    $\dir_2$ (and, in fact, their order changes between directions $\dir_1$ and~$\dir_2$).
    We consult the oracle $\mathcal{O}$ to obtain the APDs  $\dgm{}{\dir_1}$ and
            $\dgm{}{\dir_2}$.  Note that we need the zeroth- and
            first-dimensional diagrams only, and these are the only non-trivial
            diagrams for a graph.
    Let $\mathbb{B}$ be the bow tie between~$\pLine{\dir_1 }{\hFiltFun{\dir_1}(v)}$
and~$\pLine{\dir_2}{\hFiltFun{\dir_2}(v)}$.
    Note that, by construction,~$\mathbb{B}$ contains exactly one
    point from~$V$, namely~$v'$.
    This first step of the edge test takes $\Theta(\persComp)$ time and will use
    two augmented persistence~diagrams.

    The second step of the edge test is to compute indegrees of $v$ in order to determine if there exists
    an edge between~$v$ and $v'$
    (Lines~\ref{line:edge:indegreestart}--\ref{line:edge:edgeTestEnd}
    of \algref{edge_recon}).
    By \lemref{Indegree}, we compute the indegrees $\indeg{v}{\dir_1}$
    and $\indeg{v}{\dir_2}$ from~$\dgm{}{\dir_1}$ and $\dgm{}{\dir_2}$,
    respectively, in $\Theta(n)$ time;
    see Lines~\ref{line:edge:indeg1} and~\ref{line:edge:indeg2}.
    Then, using~\lemref{edgeExist}, we determine
    whether the edge~$(v,v')$ is in $E$ by checking if~$|\indeg{v}{\dir_1}-\indeg{v}{\dir_2}|=1$.
    If this equality holds, the edge
    exists; if not, the edge does not; see
    Lines~\ref{line:edge:indeg}--\ref{line:edge:edgeTestEnd}.
    The bottleneck of the edge test is the $\Theta(n)$ indegree computation,
    the second step of the edge test takes~$\Theta(n)$ time.
    We do not compute additional diagrams in this step.

    We apply the edge test for all ${n \choose 2} = \frac{1}{2}(n^2-n)$ distinct pairs in $V$.
    For all pairs, the complexity
    of the edge test uses~$n^2-n$ persistence diagrams and takes
    $\Theta(n^2\persComp + n^3)$ time.  Observing that the time to compute a
    diagram $D$ is~$\Omega(|D|)$ and using \rcorref{PDpoints}, we observe that $\persComp$ is
    $\Omega(n)$. As a result, we can simplify $\Theta(n^2\persComp + n^3)$ to
    $\Theta(n^2 \persComp)$.
    Thus, the runtime of \algref{edge_recon} is~$\Theta(n^2\persComp)$ ($\Theta(n^2\persComp)$
    for preprocessing
    and $\Theta(n^3)$ for the edge~tests).
\end{proof}
\begin{algorithm}
\caption{Reconstruct Edges}
\label{alg:edge_recon}
\begin{algorithmic}[1]
    \REQUIRE Oracle $\mathcal{O}$ for an unknown graph $G=(V,E) \subset \R^2$;
    the vertex set $V$
    \ENSURE the edge set $E$
    \STATE $E \gets \emptyset$\label{line:edge:firstline}
    \STATE $c \gets$ use~\cite{verma2011slow} to compute cyclic orderings
    for all vertices in $V$ stored
    as lists, indexed by $v\in V$\label{line:edge:1minang}
    \FOR{$v \in V$}
    \STATE $\ell[v] \gets$ cyclic ordering of lines though $v$, computed from
    $c[v]$
    \STATE $\theta(v) \gets$ minimum angle between any two lines in
    $\ell[v]$\label{line:edge:2minang}
    \ENDFOR
    \STATE $\theta = \frac{1}{2}\min_{v\in V} \theta(v)$\label{line:edge:3minang}
    \FOR{$(v,v') \in V \times V, v\neq v'$}\label{line:edge:forloopstart}
        \STATE $\dir \gets$ unit vector perpendicular to $(v' -
        v)$\label{line:edge:edgeTestStart}\label{line:edge:bowtiestart}
        \STATE $\dir_1 \gets$ $\dir$ rotated by $\theta$
        \STATE $\dir_2 \gets$ $\dir$ rotated by
        $-\theta$\label{line:edge:bowtie2}
        \STATE Consult $\mathcal{O}$ to obtain
            $\dgm{}{\dir_1}$ and
            $\dgm{}{\dir_2}$ restricted to zero- and one-dimensional points
            \label{line:edge:bowtieend}
        \STATE Compute $\indeg{v}{\dir_1}$
            from $\dgm{}{\dir_1}$
            \label{line:edge:indegreestart}\label{line:edge:indeg1}
        \STATE Compute $\indeg{v}{\dir_2}$ from  $\dgm{}{\dir_2}$\label{line:edge:indeg2}
        \IF{$|\indeg{v}{\dir_1} - \indeg{v}{\dir_2}| == 1$}\label{line:edge:indeg}
            \STATE Add $(v,v')$ to $E$ \label{line:edge:edgeCondition}
        \ENDIF\label{line:edge:edgeTestEnd}
    \ENDFOR\label{line:edge:forloopend}
    \RETURN $E$
\end{algorithmic}
\end{algorithm}

Putting together \thmref{intComp} and \thmref{edgeEmbed} leads us to our primary
result for plane graphs:

\begin{theorem}[Plane Graph Reconstruction]
\label{thm:reconstruction}
    Let~$\simComp=(V,E)$ be a plane graph with $n$ vertices embedded in $\R^2$.
    \algref{vert_recon} and \algref{edge_recon} calculate the vertex locations
    and edges using~$n^2 - n + 3$ different
    directional augmented persistence diagrams
    in $\Theta(n^2\persComp)$ time,
    where~$\Theta(\persComp)$ is the time
    complexity of computing a single diagram.
\end{theorem}

\begin{proof}
    By \thmref{intComp}, \algref{vert_recon} reconstructs the vertices $V$ using
three APDs in $\Theta(n\log{n} + \persComp)$ time.
    By \thmref{edgeEmbed}, \algref{edge_recon} reconstructs the edges $E$
with $n^2-n$ directional augmented persistence
diagrams in $\Theta(n^2\persComp)$ time.
Thus, we can
reconstruct all vertex locations and edges of $G$
using $n^2 - n + 3$ augmented persistence diagrams
in $\Theta(n^2\persComp)$ time.
\end{proof}

\subsection{Edge Reconstruction in $\R^d$}

We can also reconstruct edges of graphs embedded in higher
dimensions. We can form a higher-dimensional version of
the bow tie, referred to as a \emph{wedge}, which is the symmetric difference of two
$(d-1)$-dimensional~hyperplanes.

\begin{theorem}[Edge Reconstruction in Higher Dimensions]
\label{thm:highDEdgeEmbed}
    Let $G=(V,E)$ be a straight-line embedded graph in $\R^d$ for some~$d>1$.
If~$V$ is known, then we can compute~$E$ using $n^2 - n$ directional
augmented persistence~diagrams in~$O(n^2 \persComp  + n^4)$ time,
    where~$\Theta(\persComp)$ is the time
    complexity of computing a single
    diagram.
\end{theorem}
\begin{proof}
    Let $\subspace{e_1}{e_2}$ be the subspace of $\R^d$
    spanned by $e_1$ and $e_2$.  Let $\pi \colon \R^d \to \R^2$ be defined by
    $\pi(x_1,x_2, \ldots, x_d) = (x_1,x_2)$; in other words,~$\pi$
    is the projection to~$\subspace{e_1}{e_2}$. Let $V_* = \pi(V)$ and
    note that by the \assumptref{gp}, no three vertices are collinear in
	$V_*$ and no two points share any coordinate values.
    Since the vertex set $V_*$ lies in a two-dimensional plane, we can use
    Lines~\ref{line:edge:firstline}--\ref{line:edge:3minang} of
    \algref{edge_recon} to compute an angle $\theta$ for $V_*$.

    Let $v,v' \in V$. We use a method similar to the one defined in
    Lines~\ref{line:edge:bowtiestart}--\ref{line:edge:bowtie2} of
    \algref{edge_recon} to find an appropriate wedge to test whether an
	edge between $v$ and $v'$ exists. We define $\dir_1$ and $\dir_2$ such that the
    lines $\simplePLine{1}{\pi(v)}$ and $\simplePLine{2}{\pi(v)}$, which are
    perpendicular to $\dir_1$ and~$\dir_2$ and go through $\pi(v)$, define a bow
    tie at $\pi(v)$ that isolates the potential edge~$(\pi(v), \pi(v'))$.  This bow tie
    extends to a wedge in $\R^d$ by replacing the lines with hyperplanes that
    intersect $\subspace{e_1}{e_2}$ orthogonally; specifically,
    the line $a_1x_1+a_2x_2=h$ in
    $\subspace{e_1}{e_2}$
    corresponds to
    the~$(d-1)$-dimensional hyperplane,  $a_1x_1+a_2x_2+0x_3+...+0x_d=h$
    in~$\R^d$. Let $\dirLifted_1$ and $\dirLifted_2$ be directions in $\sph^{d-1}$ that
    define this wedge.
	Because $\pi$ is an orthogonal projection and points in~$V$ satisfy the
	\assumptref{gp},~$v'$ must be the only vertex in
	this wedge, thus it isolates the edge $(v,v')$, should it exist.

    We then compute the indegrees of $v$ with respect to $\dirLifted_1$
    and $\dirLifted_2$,
    just as we did in the
    two-dimensional case in Lines~\ref{line:edge:indeg1}
    and~\ref{line:edge:indeg2} of \algref{edge_recon}.  However, we note that
    since our dimension may be greater than two,~\lemref{Indegree} states
    that this step takes $O(n^2)$ time for each indegree computation. We
    perform indegree checks on ${n \choose 2}$ pairs of
    vertices.  Finally, by~\lemref{edgeExist}, we test for an
    edge by determining if the difference
    between the indegrees is one using the same technique as
    Lines~\ref{line:edge:indeg}-\ref{line:edge:edgeCondition} of \algref{edge_recon}.

    The runtime for computing the indegree for all ${n\choose 2}$ pairs of
    vertices and reconstructing the edges is $O(n^2 \persComp + n^4)$.
    Similar to \thmref{edgeEmbed}, the reconstruction uses two diagrams for
    each pair of vertices, which is $n^2 - n$ diagrams.

\end{proof}

Putting together \thmref{genIntComp} and \thmref{highDEdgeEmbed}, we obtain a
graph reconstruction algorithm for graphs embedded with straight-line edges in
$\R^d$:

\begin{theorem}[Generalized Graph Reconstruction]
    If~$\simComp=(V,E)$ is a straight-line embedded graph in $\R^d$ for~$d>1$, then
$V$, $E$, and the embedding of~$G$ can be computed
    using augmented persistence diagrams \mbox{from~$n^2 -n + d + 1$} different directions in
    time $O(dn^{d+1} + n^4 + (d + n^2)\persComp )$,
    where~$\Theta(\persComp)$ is the time
    complexity of computing a single~diagram.
\label{thm:edge-higher-d}
\end{theorem}

\section{Experiments}
\label{sec:experiments}

In this section, we empirically validate both the correctness of our algorithm
and the runtime of various steps of the algorithm.
Furthermore, we explore the runtime of checking for the presence of edges
for a fixed vertex set as the number of edges varies.
We conclude with an analysis of the probability of having `bad
input;' that is, input with small angles.
Timings were taken on a MacBook Pro with a 2.3GHz Intel Core i5 processor,
8GB RAM,
running macOS Mojave~10.14.6.
Our implementations were written in C++, compiled with the LLVM toolchain version 10.0.1.
Our code is publicly available and experiments were run using the reconstruction library at
git commit hash \emph{d9fd610}\footnote{The library is available at
\url{https://github.com/compTAG/reconstruction}.}.

\paragraph{Implementation}

To complement the theoretical results of this paper, we have written the reconstruction code in C++.
In the implementation, an
oracle is given a direction and returns an augmented persistence diagram.  We use the
Dionysus~2.0 library~\cite{dionysus} for computing
the augmented persistence diagram.  We store the computed birth-death pairs in two lists,
one for zero-dimension points (representing connected components) and one for
one-dimensional points (representing loops).

We note that we have made some changes from the graph reconstruction algorithm
presented in the previous sections.
In particular, instead of using \cite{verma2011slow} to compute a cyclic ordering of all vertices in
$\Theta(n^2)$, in \algref{edge_recon} \lineref{edge:1minang},
we implemented a na\"ive approach of
ordering around each vertex independently in~$O(n^2 \log n)$.

\paragraph{Data}

Our experimental data is a set of random plane graphs embedded in $\R^2$.
For~$n\in \Z^+$,
we randomly generate~$n$ points from the uniform
distribution over the unit square. Next, we compute the Delaunay triangulation using
the SciPy library~\cite{scipy} in Python. This random triangulation is similar to those
of~\cite{boots1999spatial, chenavier2016stretch, devillers2018walking}, but uses a
binomial point process with a uniform density rather than a Poisson point process.
Finally, we arrive at our random
graph by setting a percentage~$\alpha$ of edges~$E$ from the triangulation to keep and
delete edges until~$\alpha \cdot |E|$ remain. The result is a random subgraph
of a Delaunay~triangulation.

\paragraph{Timing}

For each experiment about timing, we subtract out the time spent computing
diagrams as to capture only the time spent in the algorithm.  In
addition, we run each experiment five times and report the average of the runs.

\subsection{Experimental Runtimes}
Our first experiment validates the theoretical runtime of the plane graph
reconstruction algorithm.  We
fix the percent of edges at~$\alpha = 10\%$ and
vary the number of vertices~$n \in \{10,20,\ldots,80\}$.
For each value of $n$, we reconstruct ten random graphs as described in
\pararef{Data}{experiments}.
We track the time spent in
vertex and edge reconstruction
and record the timings as described in \pararef{Timing}{experiments}.

In \figref{vertex_runtimes} and \figref{edge_runtimes}, we plot the time
for reconstructing vertices and edges versus $n$, respectively.
Times are measured in milliseconds
and each mark is the mean for five runs on each random graph.
In \figref{vertex_runtimes} we fit~\mbox{$t = \beta_0 + \beta_1 n \log n$}
to the experimental data, which had a RMSE of 0.0019 ms.
The figure agrees with \thmref{intComp}, which showed that the
vertex reconstruction runtime is $O(n\log n)$
(ignoring the time for computing diagrams).
\figref{edge_runtimes} must be interpreted with a little more care.  Recall that
\thmref{edgeEmbed} showed that the edge reconstruction
runtime is~$\Theta(n^2\persComp)$. However, the runtime of the loop in
Lines~\ref{line:edge:forloopstart}--\ref{line:edge:forloopend}
of \algref{edge_recon} is $O(n^3)$, and in the experiments,
we subtract the time for computing diagrams.
Thus, we expect the time to be upper-bounded by $n^3$.
We fit the curve $t = \beta_0+\beta_1 n^3$ to the experimental data,
which had a RMSE of 0.45 ms.
The figure agrees with \thmref{edgeEmbed}.

\begin{figure}[htb]
    \centering
    \begin{subfigure}[b]{0.47\textwidth}
        \includegraphics[width=\textwidth]{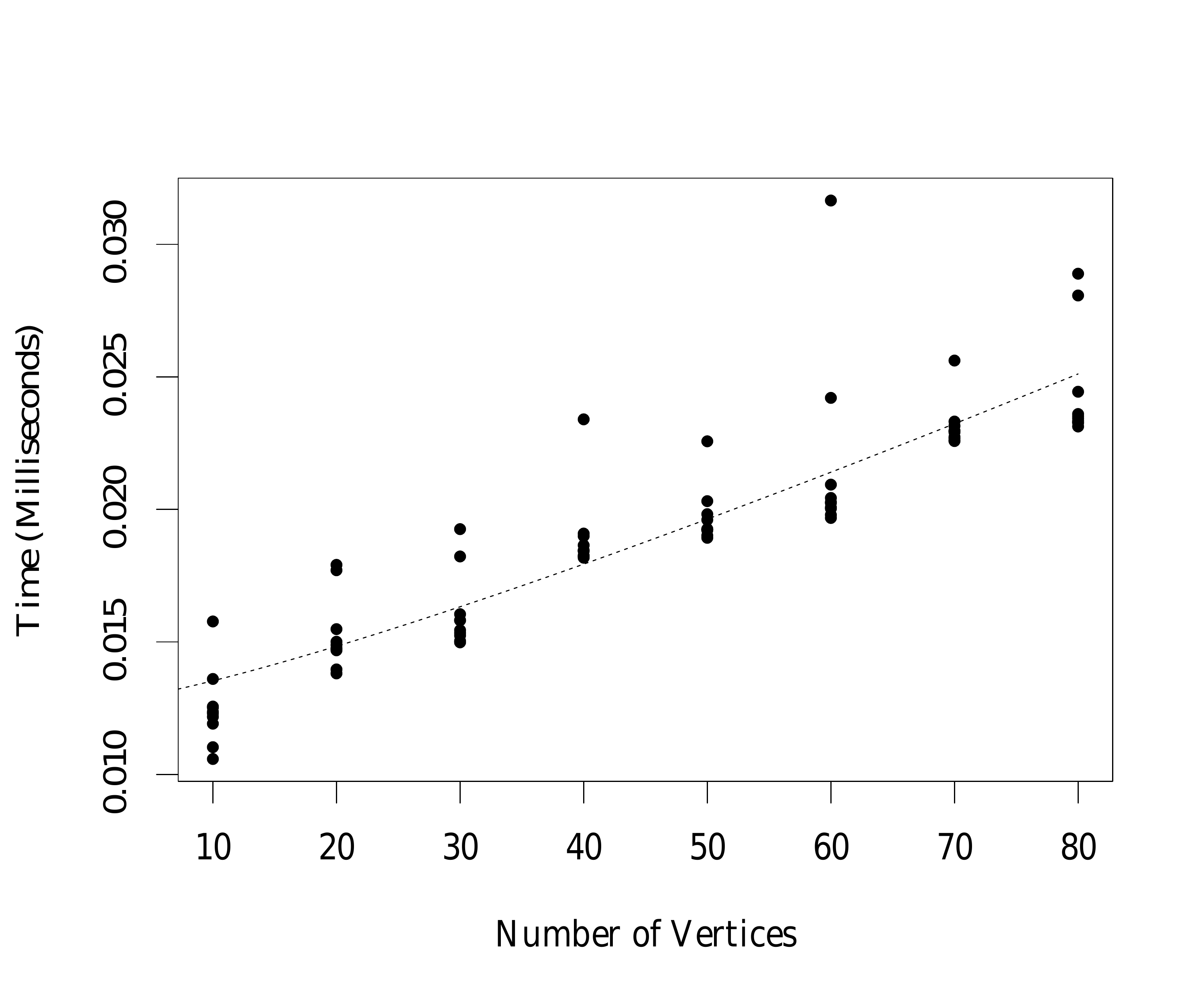}
        \caption{Vertex Reconstruction}
        \label{fig:vertex_runtimes}
    \end{subfigure}
    \quad
    \begin{subfigure}[b]{0.47\textwidth}
        \includegraphics[width=\textwidth]{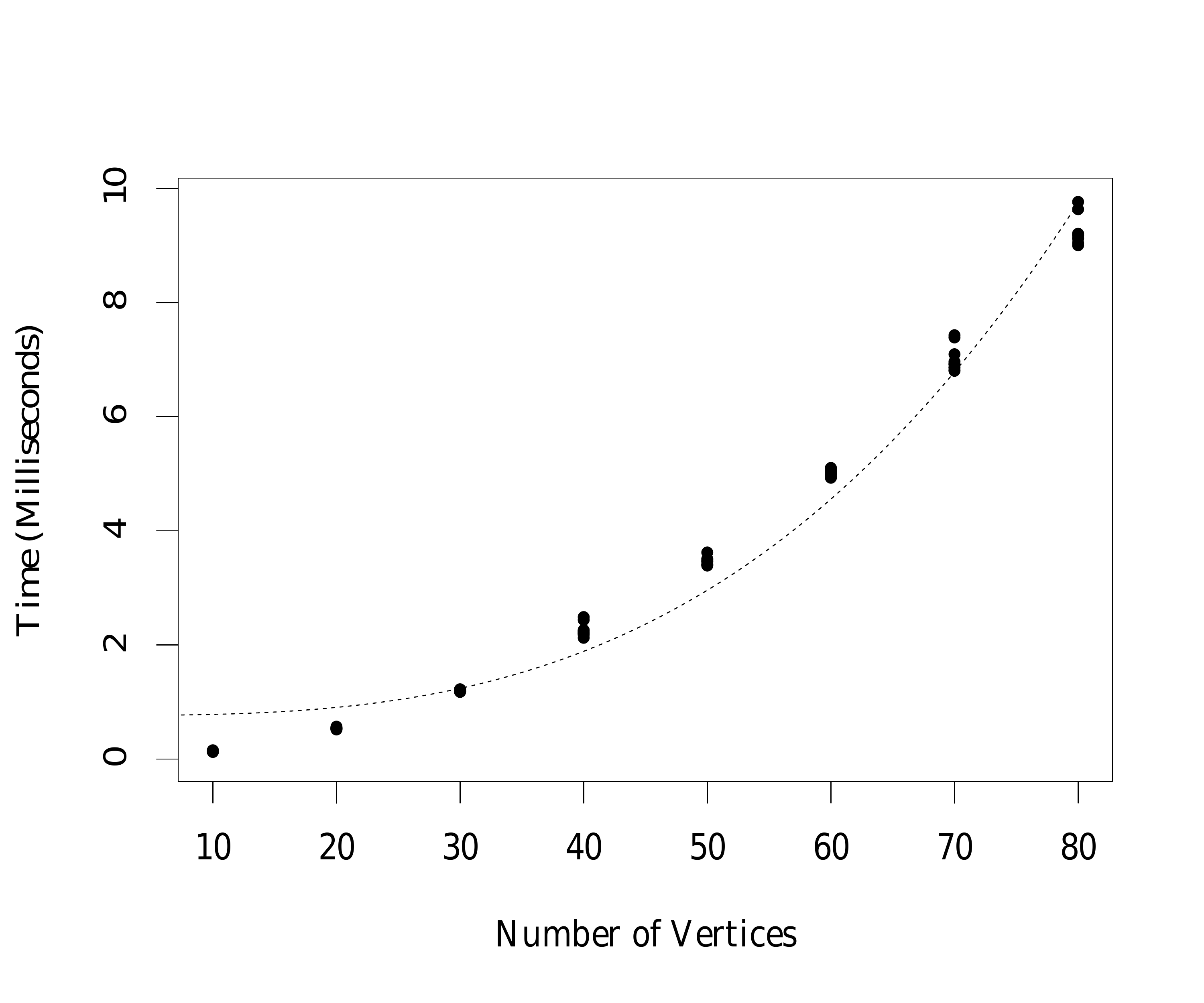}
        \caption{Edge Reconstruction}
        \label{fig:edge_runtimes}
    \end{subfigure}
    \caption{Experimental times for reconstructing plane graphs. For each
        value of $n$, we record the vertex and edge reconstruction times
        for $10$ random graphs and compare to the theoretical runtimes
        for vertex reconstruction (\thmref{intComp})
        and edge reconstruction (\thmref{edgeEmbed}).
        The experimental timings agree with the theoretical~runtimes.}\label{fig:runtimes}
\end{figure}

\subsection{Effect of Edges on Runtime}
\label{cpedge}

Our second experiment investigates the effect of edge density on our algorithm.
Note that while our runtimes are expressed in the number of vertices,
investigating the edges adds an additional insight.
As in the previous section, each graph is generated as described in
\pararef{Data}{experiments} and we track the time spent in the vertex
and edge reconstruction as described in \pararef{Timing}{experiments}.
Specifically, for a fixed Delaunay triangulation on a vertex set of size $n$,
we vary the percent of edges $\alpha$ from the Delaunay triangulation.
We use~$\alpha \in \{10,20,\ldots, 100\}$ and~\mbox{$n \in \{10, 25, 50, 100\}$.}
In \figref{vertex_vs_edges} and \figref{edges_vs_edges}, each value of $n$ is
represented by a different line in the plot.

In the first part of this experiment, we investigate
the effect of edge density on vertex reconstruction times.
In \figref{edges_vs_edges}, we see a constant relationship in the average
runtime per vertex as a function of the percent of edges.
To confirm, we fit the curve $t= \beta_0$,
with RMSE
218.5 ns for $n=10$,
62.4 ns for $n=25$,
33.4 ns for $n=50$, and
18.2 ns for $n=100$.
Thus, as predicted by the analysis of \thmref{intComp}, the runtime of the
vertex reconstruction algorithm is independent of the number of edges.
\begin{figure}[htb]
  \centering
  \begin{subfigure}[b]{0.47\textwidth}
    \includegraphics[width=\textwidth]{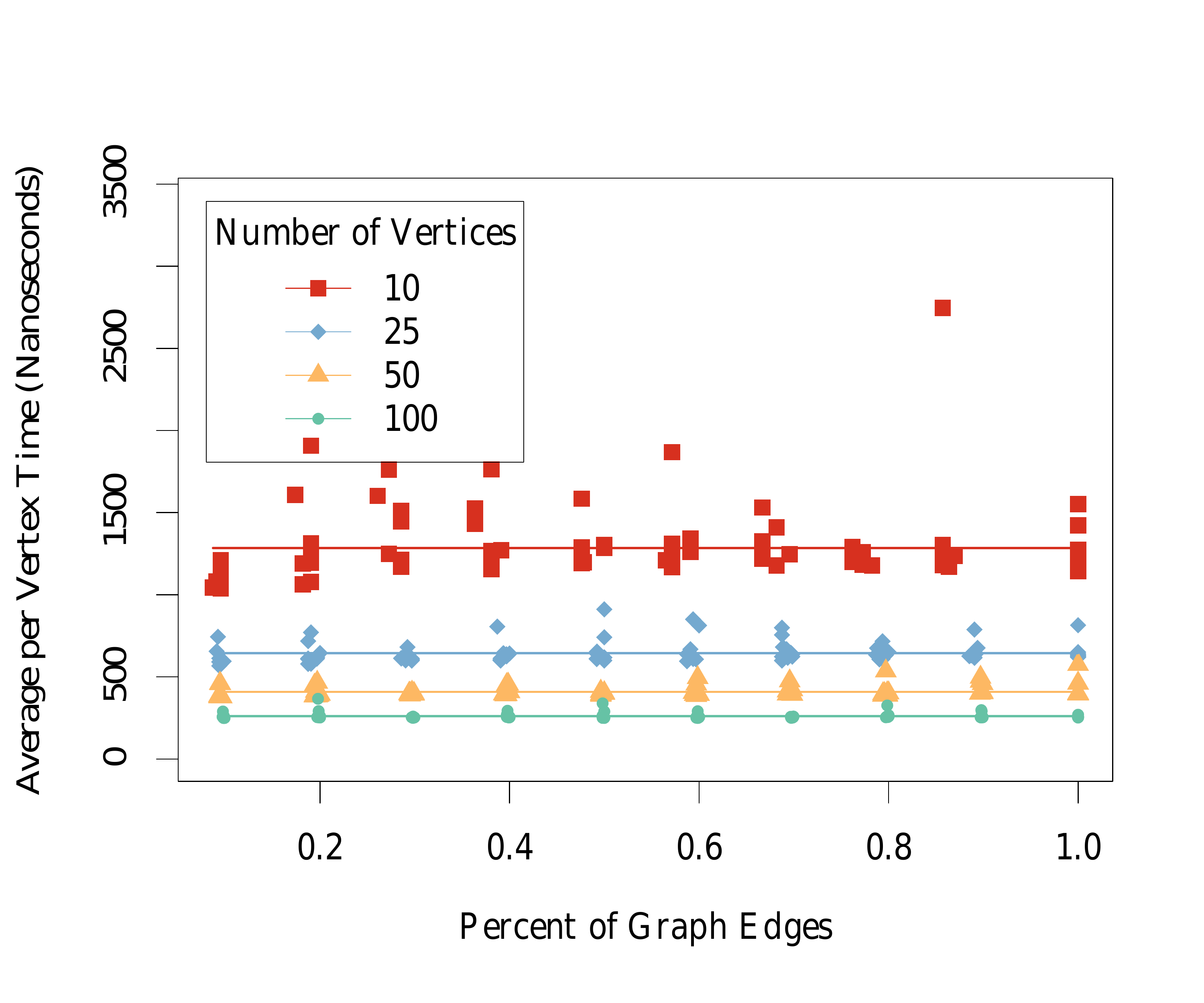}
    \caption{Vertex Time}
    \label{fig:vertex_vs_edges}
  \end{subfigure}
  \quad
  \begin{subfigure}[b]{0.47\textwidth}
    \includegraphics[width=\textwidth]{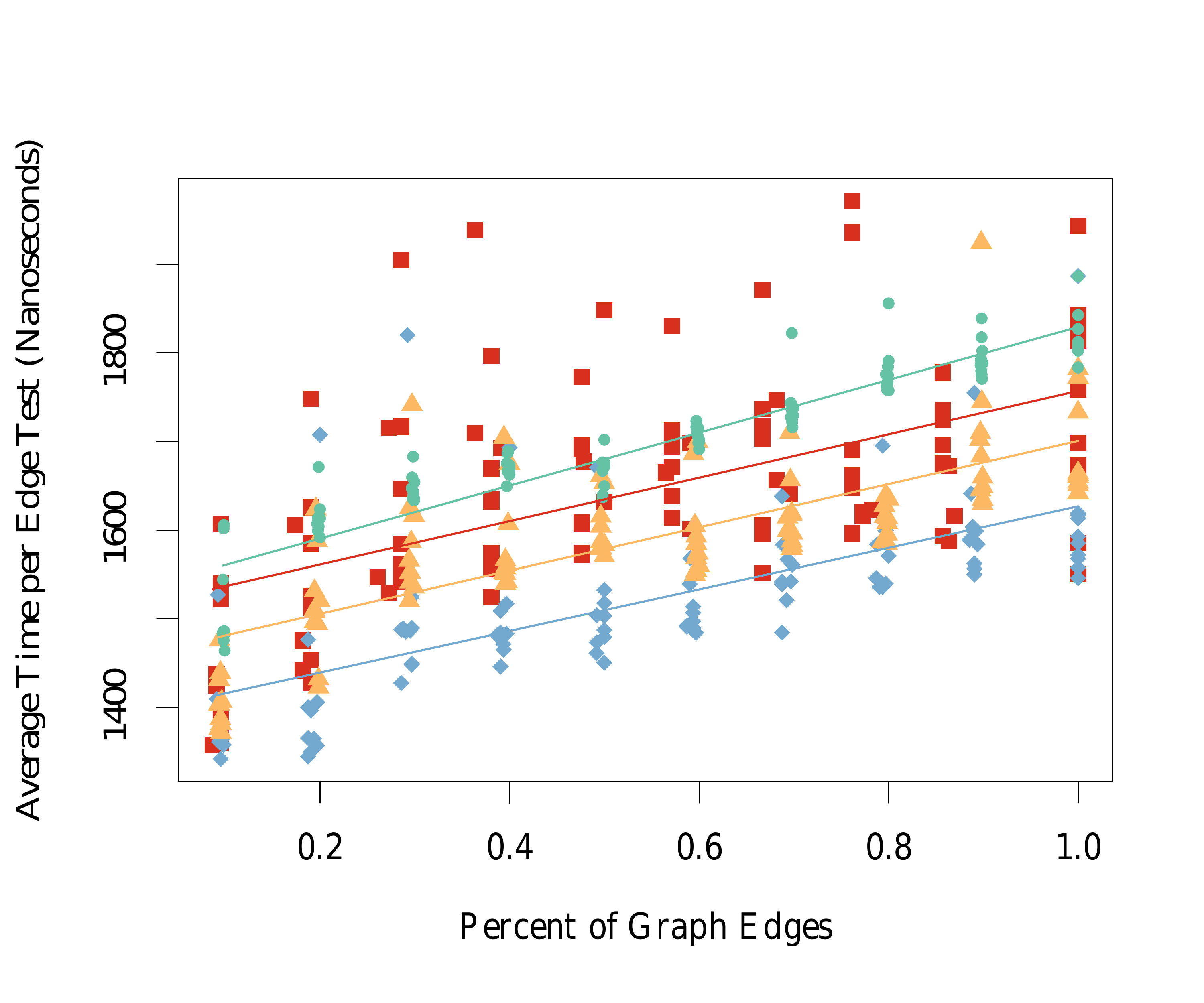}
    \caption{Edge Time}
    \label{fig:edges_vs_edges}
  \end{subfigure}
  \caption{
      Experimental reconstruction times for plane graphs as the number of
      edges increases.
      For each $n$, we record the vertex and edge reconstruction times for $10$
      random graphs as we increase the percent of edges from a
      Delaunay triangulation.
      \emph{Left:} From (\thmref{intComp}), we expect that the vertex reconstruction
      time is independent of the number of edges.
      \emph{Right:} From (\thmref{edgeEmbed}), for a graph with $n$ vertices
      and $m$ edges, we conduct $n^2-n$ edge tests on a given graph.
      Each edge test compares two augmented persistence diagrams of size $O(n+m)$.
      Thus, we expect that the time grows linearly with respect to the number
      (and hence the percent) of edges.
      In both plots, the experimental timings agree with the theoretical run times.
  }
  \label{fig:vs_edges}
\end{figure}

In the second part of this experiment, we focus on the reconstruction time of
edges in the graph.
In \figref{edges_vs_edges}, we see a linear relationship in the average
runtime per edge as a function of the percent of edges.
To confirm, we fit the curve $t = \beta_0+\beta_1 n$,
with RMSE
$108.3$ ns for $n=10$,
$72.9$ ns for $n=25$,
$61.9$ ns for~$n=50$, and
$33.4$ ns for $n=100$.
The linear growth is somewhat unexpected as the number of edges does not appear
in the analysis of \thmref{edgeEmbed}.
But, recall that the edge reconstruction algorithm computes indegree.
By \lemref{Indegree}, the computation takes time proportional
to the size of the augmented persistence diagrams (in dimensions zero and
one combined).
Moreover, the size of the one dimensional diagram grows linearly with respect to
the number of edges.
Thus, when~$\alpha$ is small, as would be the case in sparse graphs,
we see fewer cycles, which makes $\dgm{1}{\cdot}$ smaller. This results in
a small improvement in runtime.

\subsection{Minimum Angle}\label{ss:minangle}

Very small angles cause numerical issues, which is a problem that we encountered
in the above experiments.
In order to mitigate the issues, our code has an assertion that bow tie
half-angles are at least $10^{-6}$ radians. We observed that many of the experiments with
over $50$ vertices were failing this assertion.  In addition, a bounded angle
assumption was used in~\cite{curry2018directions}.
Thus, we investigate the
probability of encountering small angles, both with a back-of-the-envelope calculation
and empirical observations.

\paragraph{Back-of-the-Envelope Calculation}
Let
$a,b,$ and $c$ be three points sampled i.i.d.\ from the uniform distribution
on the unit square.  We consider $\angle abc$.  Without loss
of generality, we can translate and rotate the points such that $b$ is at the origin,
and $c$ is on the positive $x$-axis.  Then, consider the line $ba$.  If $a$ was randomly
chosen, then any angle it makes with the positive $x$-axis ($ba$) is equally likely, so $$ \mathbb{P}(\angle abc
< 10^{-6}) = \frac{2\cdot10^{-6}}{2\pi} = \frac{10^{-6}}{\pi} \approx
3.18e{-7}.$$
Furthermore, let $S_n$ be a set of $n$ (different)
points sampled i.i.d.\ from the uniform distribution on the unit
square.
Notice that we can make \mbox{${{n}\choose{3}} = n(n-1)(n-2)$} angles defined by the
points in $S_n$.
Let $A_n$ be the event that there exists an angle less than $10^{-6}$ in $S_n$.
Assuming independence of angles in $S$, we have:
$$
\mathbb{P}(A_n) = 1 - \left(1-\frac{10^{-6}}{\pi}\right)^{n(n-1)(n-2)}.
$$
We observe that $\mathbb{P}(A_n)>5\%$ when $n \geq 56$.  In the above runtime
experiments, it is not surprising that we see the assertion failing.
In this calculation, we assume that all angles are independent, so
$\mathbb{P}(A_n)$ is an overestimation of the true probability. We conjecture:

\begin{conjecture}[Probability of Encountering Small Angles]
    Let $G=(V,E)$ be a randomly generated plane graph embedded in $\R^2$.
    If $|V|>55$, then with probability at least $5\%$, the minimum angle as
described in~\thmref{edgeEmbed} is less than $10^{-6}$ radians.
\end{conjecture}

Our final experiment is an investigation of the minimum angle of random point
sets. This is motivated by understanding the frequency of small angles, and the
extent to which our assertion and the assumption of~\cite{curry2018directions}
may or may not be limiting in practice.

\paragraph{Empirical Observations}
This
experiment illustrates how the bounded angle assumption is actually quite limiting, and that
small angles appear with high probability as the number of vertices increases.

\begin{figure}[htb]
\begin{center}
\includegraphics[scale=0.3]{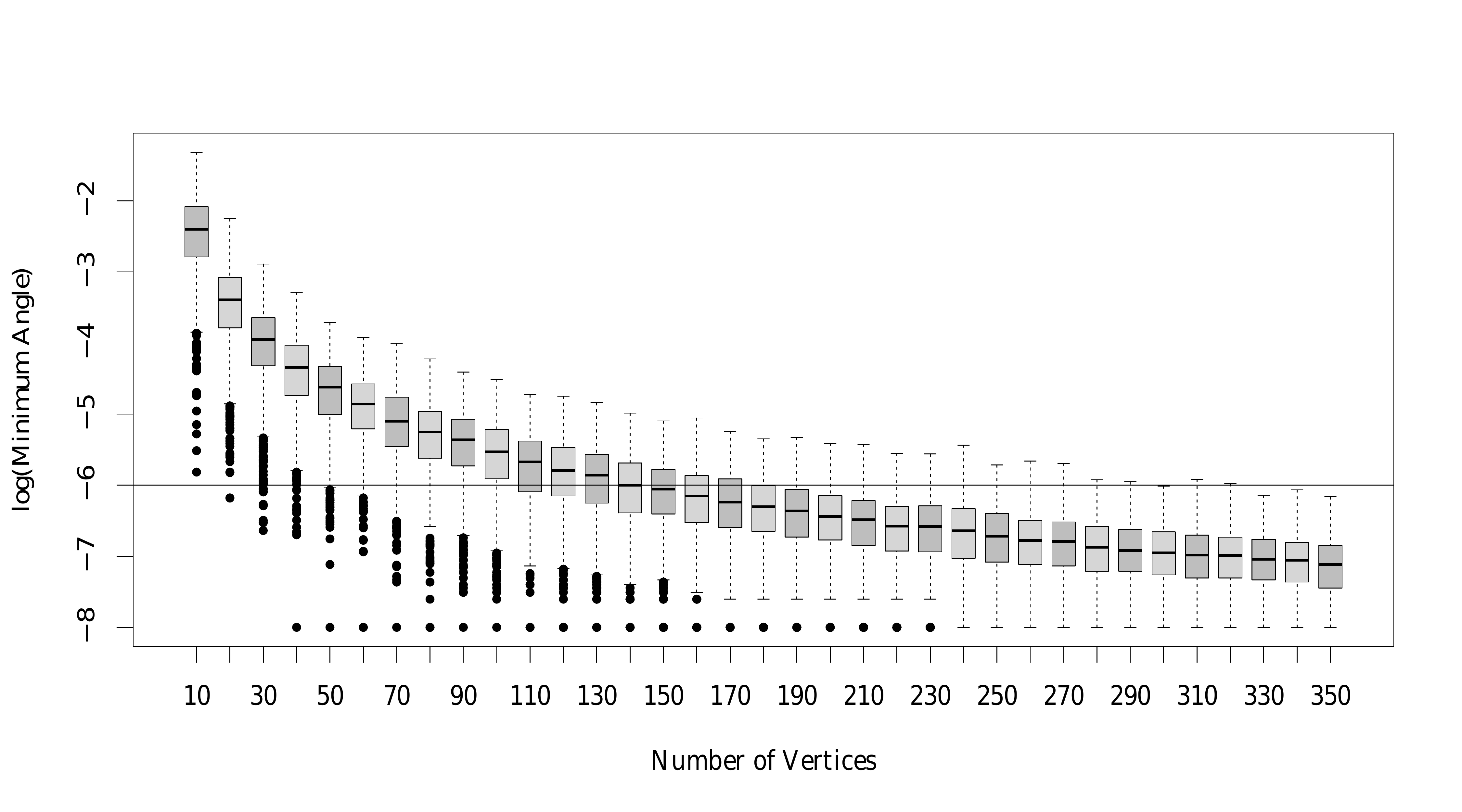}
\caption{Box plots for distributions of minimum angles for $100$
random graphs with $n$ vertices. Any graph that has a minimum angle less than
 $10^{-6}$, denoted by the horizontal line, encounters numerical errors from the
reconstruction algorithm. As the number of vertices increases, the number of
graphs with a minimum angle of $10^{-6}$, and thus encountering numerical errors,
increases.}
\label{fig:min_angle}
\end{center}
\end{figure}
We generated random sets of vertices of size $n\in \{10, 20, \ldots, 350\}$,
and measured the minimum angle between all triples of vertices.
In \figref{min_angle}, we show the box plots of the minimum
angles for $1000$ random graphs for each~$n$.
In \figref{prob_failure},
we show the probability of a graph having a minimum angle less than~$10^{-6}$
as predicted by the theoretical model and by the experimental data.
We notice at $n=20$, we begin to see
minimum angles of $10^{-6}$ radians appear in our random graphs. At~$n=70$, the
number of random graphs with a minimum angle of $10^{-6}$ radians or less
becomes even more substantial, with~$8\%$ of the angles less than~$10^{-6}$
experimentally.
This experiment suggests that for plane graphs with~\mbox{$n \geq 50$,} our algorithm
may encounter numerical errors due to small~angles.  And, as $n$ grows into the
hundreds, the probability gets close to one.  For the
range of $n$ values we tested experimentally, we see---as we expected---that the
theoretical back-of-the-envelope calculation upper bounds the experimental
probability of finding a small angle.  As we mentioned above, the differences in
\figref{prob_failure} are due to the fact that the predicted probability is done
assuming independence of angle.
\begin{figure}[htb]
\begin{center}
\includegraphics[scale=0.3]{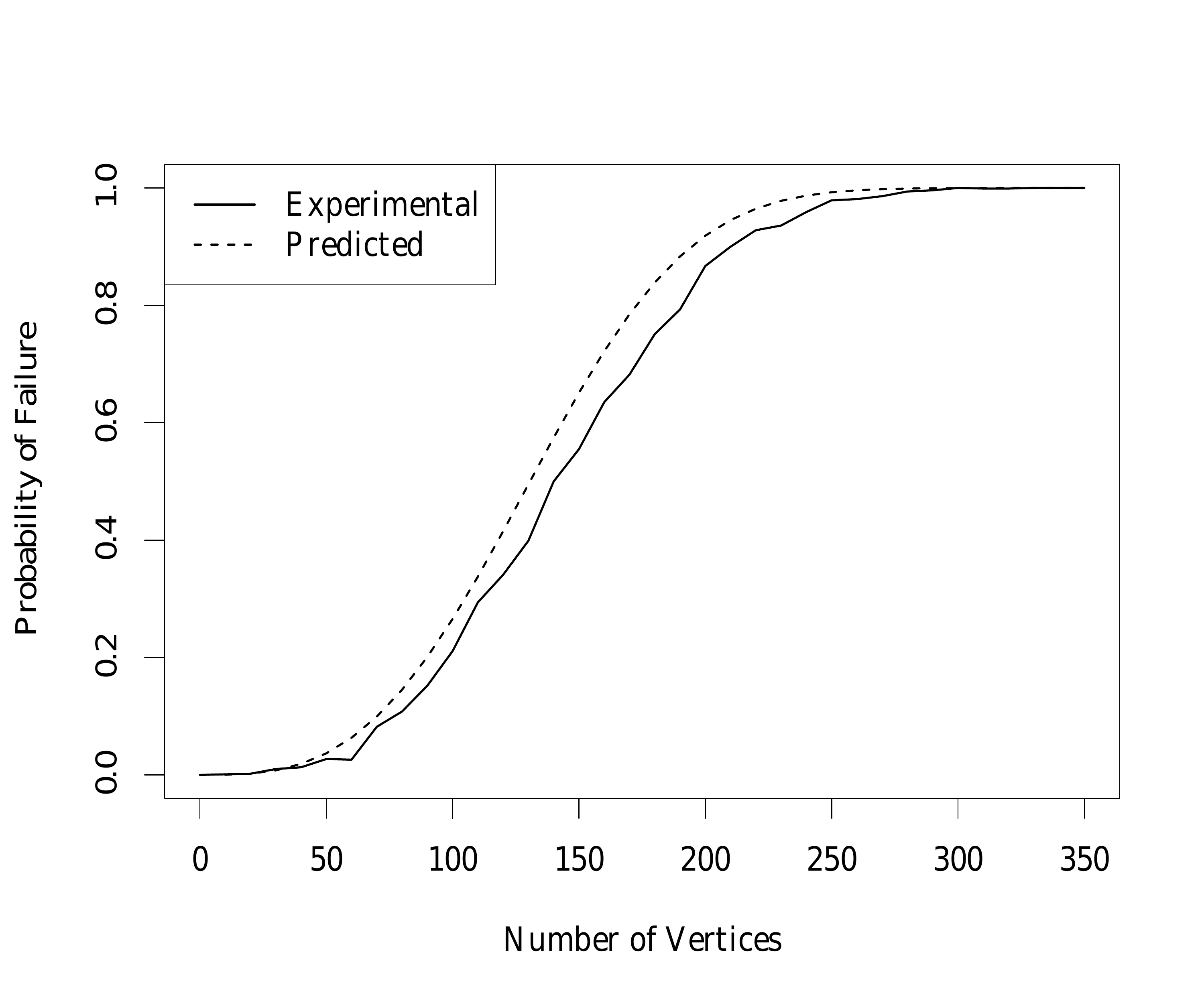}
\caption{Probability of failure (encountering an angle less than $10^{-6}$)
    from experimental data (solid line) and the theoretical model (dashed line).
The theoretical model provides an upper bound for the experimental
probability of finding a small angle.}
\label{fig:prob_failure}
\end{center}
\end{figure}

\section{Discussion}
\label{sec:discussion}

In this paper, we address the problem of reconstructing embedded graphs using persistence
diagrams through the following three main contributions.

\begin{itemize}
\item[1.] We provide the first deterministic algorithm for reconstructing a
    plane graph using (directional, augmented) persistence diagrams.  For a graph with $n$
vertices embedded in~$\R^2$, this algorithm uses $n^2-n+3$ directional APDs in
$\Theta(n^2T_G)$ time.
\item[2.] We extend the algorithm for reconstructing plane graphs to
reconstructing graphs embedded in $\R^d$ for~$d\geq 2$. This algorithm
uses $n^2-n + d + 1$ directional APDs in
$O(dn^{d+1} + n^4 + (d + n^2)\persComp )$ time.
\item[3.] We experimentally validate the correctness and time complexity of our
algorithm for reconstructing plane graphs by implementing the algorithm and
testing the implementation on randomly generated plane graphs.
We found that that the empirical probability of encountering numerical issues are
likely to occur even for graphs with only~$20$--$50$ vertices.
\end{itemize}

We identify a number of avenues for future work.
We recently extended the methods given here to
reconstruct geometric simplicial complexes (not just embedded graphs) using
APDs; see the preprint~\cite{fasy2019persistence}.  It seems possible, however,
that the algorithm for computing the generalization of indegree could be
improved.
Additionally, it would be interesting to explore reconstruction problems
with other topological descriptors,
such as Euler characteristic curves~(ECCs).
Indeed, in~\cite{fasy2018challenges}, we identified some of the challenges
in reconstructing with ECCs.  Thus,
algorithmic reconstruction with ECCs is still largely unexplored.
In other extensions of this work, one could consider variants of the inverse
problem and classify simplicial complexes, up to some topological invariant
(e.g., homotopy type)
instead of up to geometric representation. One can
hope that computing such a classification would require fewer persistence diagrams
than is required for a complete geometric reconstruction.

Finally, our algorithms rely on \emph{augmented}
persistence diagrams, but we conjecture that the corresponding persistence
diagrams (without the explicitly stored on-diagonal points) can be used to
sufficiently differentiate one shape from another.
If proven, this conjecture would close the gap between
our theoretical results for APDs and the use of PDs in practice, such as in
\cite{hofer2017deep}.

\paragraph{Acknowledgements}
This material is based upon work supported by the National Science Foundation
under the following grants: CCF 1618605 (BTF, SM, RM), DBI 1661530 (BTF, DLM,
LW), DGE 1649608 (RLB, AS), and DMS~1664858 (RLB, BTF, AS, JS).
Additionally, RM thanks the Undergraduate Scholars Program for both travel
funding and undergraduate research grants.
All authors thank the CompTaG club at Montana State University and the CCCG reviewers
for their thoughtful feedback on this work.  In particular, we thank Brad McCoy
for checking the examples for correctness. Finally, we thank the reviewers for
their careful reading and constructive comments.

\bibliographystyle{acm}
\bibliography{cccg-journal.bib}

\appendix

\section{Demonstration of Plane Graph Reconstruction}
\label{sec:Example}

We give an example of reconstructing a plane graph embedded in $\R^2$.
Consider the complex, $G$, given in \figref{exampleVertex}. The vertices of $G$
are
 $$V=\{(-1,2),(0,-1),(0.25,0),(1,1)\},$$
and edges are given by the following
pairs of vertices,
$$E=\{((-1,2),(0,-1)),((0,-1),(0.25,0)),((0,-1),(1,1)),((0.25,0),(1,1))\}.$$

\paragraph{Vertex Reconstruction}

We find vertex locations using the algorithm described in \secref{vRec}.
Note, that in this example,~$n=4$.  Using the persistence diagrams from
height filtrations in directions $e_1=(1,0)$ and $e_2=(0,1)$, we construct
the set of lines $\pLines{e_1}{V}=\{(-1,y), (0,y), (0.25,y), (1,y)\mid y\in \R\}$
and~$\pLines{e_2}{V}=\{(x,2),(x,-1),(x,0),(x,1)\mid x\in \R\}$ as shown in
\figref{Vertex-third_dir}. The set~$\pLines{e_1}{V}\cup \pLines{e_2}{V}$
of $2n=8$ lines has~$n^2 = 16$ possible locations for the vertices at the intersections
in~$A$. We show these filtration lines and intersections in \figref{Vertx-second_dir}.

We compute the third direction,~$\dir$, using the algorithm outlined in
\thmref{intComp}. Recall, that we need to find the maximum width between two lines
in $\pLines{e_1}{V}$, denoted by $w$, and smallest height between two adjacent lines
in~$\pLines{e_2}{V}$, denoted by $h$. In our example, $w=1-(-1)=2$ and
$h=2-1=1$. Then, we use the rectangle of width, $w$, and height, $h$, to choose
a direction~$\dir$. We pick $\dir$ to be a unit vector perpendicular to $[w,h/2]$.
In particular, we compute $\dir=(-0.243,0.970)\in \sph^1.$ Then, the four
three-way intersections in~$\pLines{e_1}{V} \cup \pLines{e_2}{V} \cup \pLines{\dir}{V}$
identify all Cartesian coordinates of the vertices in the graph.

\begin{figure*}
\centering
\begin{subfigure}{.3\linewidth}
  \centering
  \includegraphics[width=\linewidth,height=.8\linewidth]{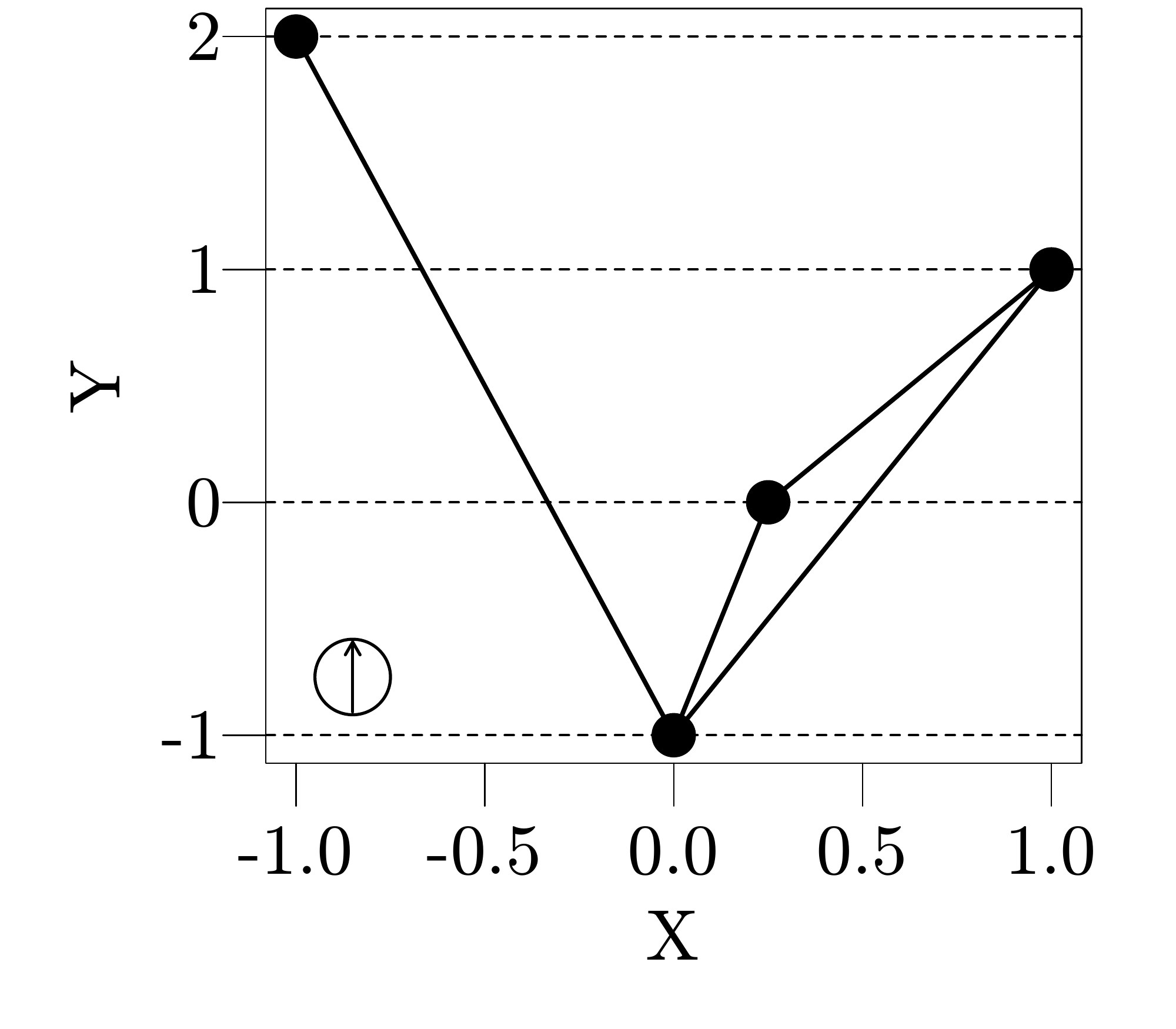}
  \caption{Filtration lines for $e_2$}
\end{subfigure}%
\begin{subfigure}{.3\linewidth}
  \centering
  \includegraphics[width=\linewidth,height=.8\linewidth]{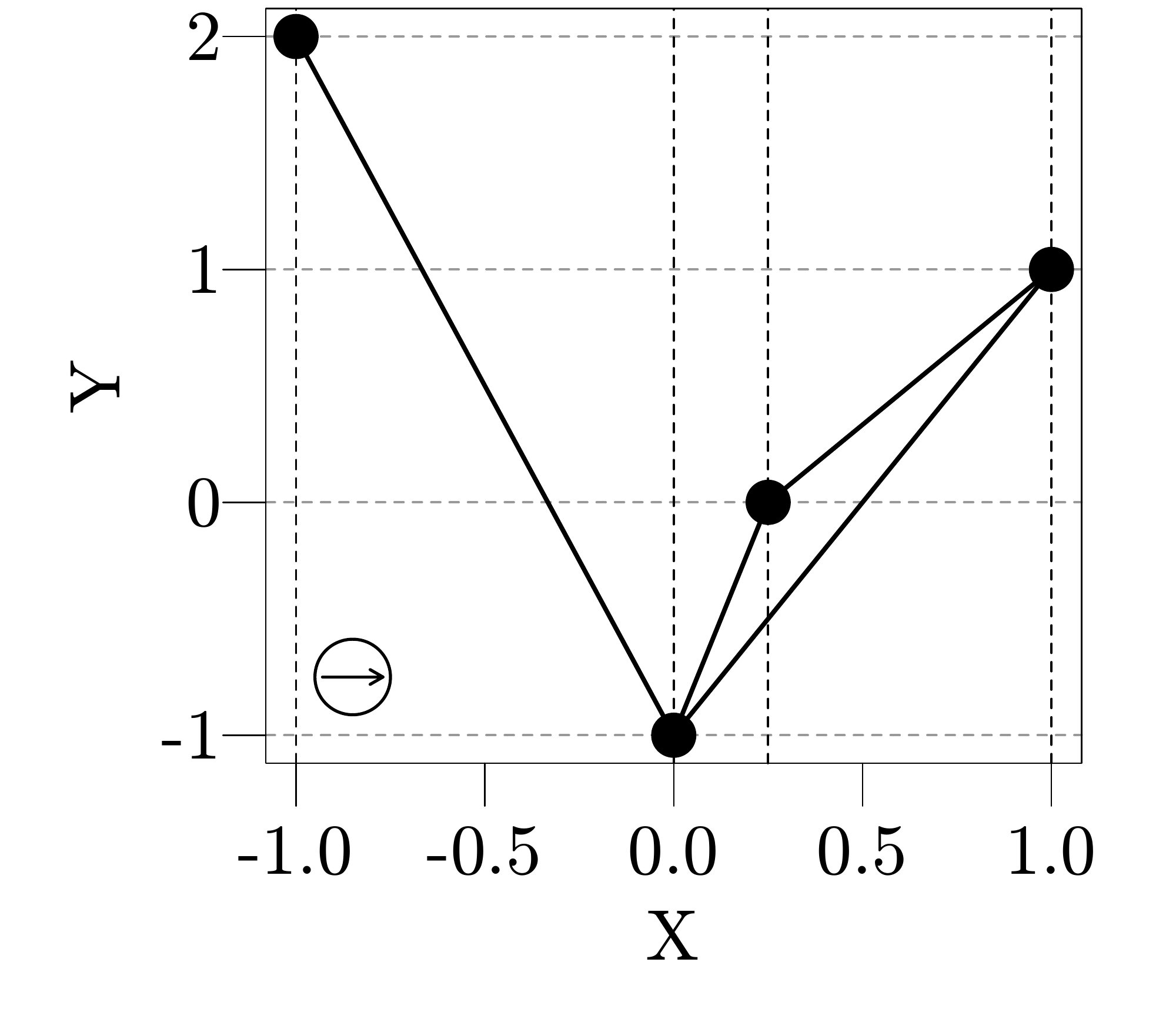}
  \caption{Filtration lines for $e_1$}
  \label{fig:Vertx-second_dir}
\end{subfigure}%
\begin{subfigure}{.3\linewidth}
  \centering
  \includegraphics[width=\linewidth,height=.8\linewidth]{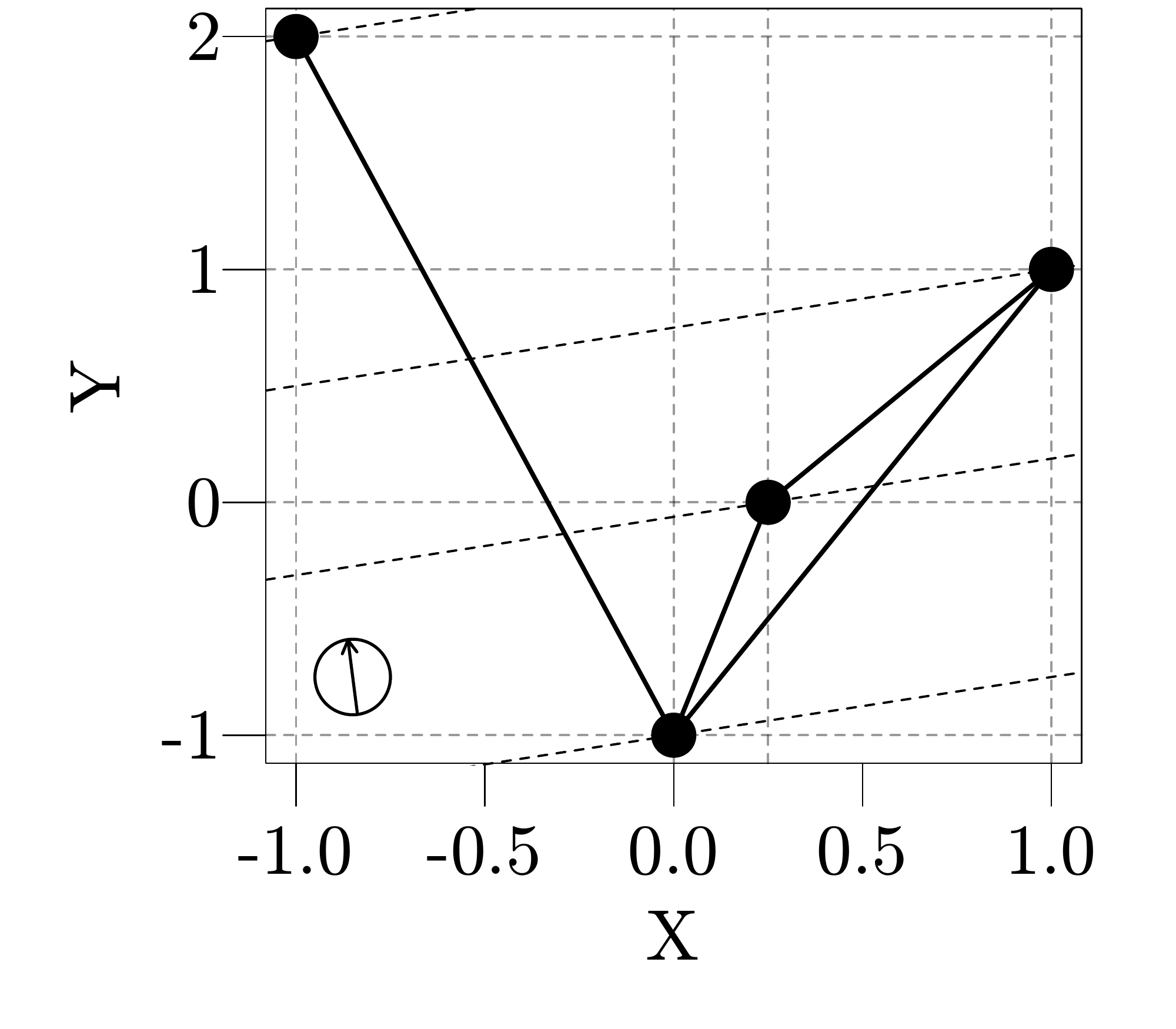}
  \caption{Filtration lines for $\dir$}
    \label{fig:Vertex-third_dir}
\end{subfigure}\\
\vspace{1cm}
\begin{subfigure}{.3\linewidth}
  \centering
  \includegraphics[width=\linewidth,height=.8\linewidth]{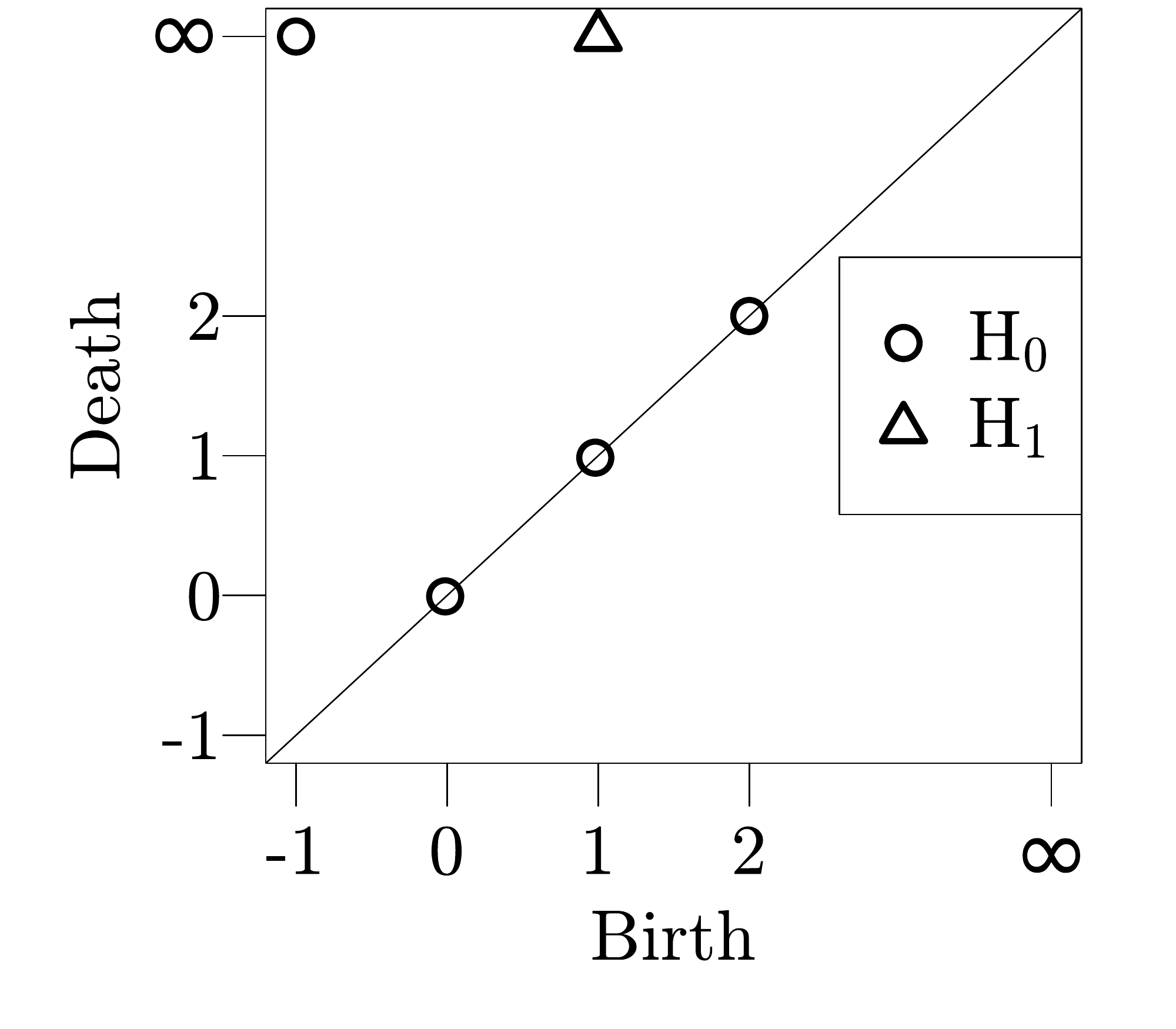}
  \caption{Diagrams for $e_2$}
\end{subfigure}
\begin{subfigure}{.3\linewidth}
  \centering
  \includegraphics[width=\linewidth,height=.8\linewidth]{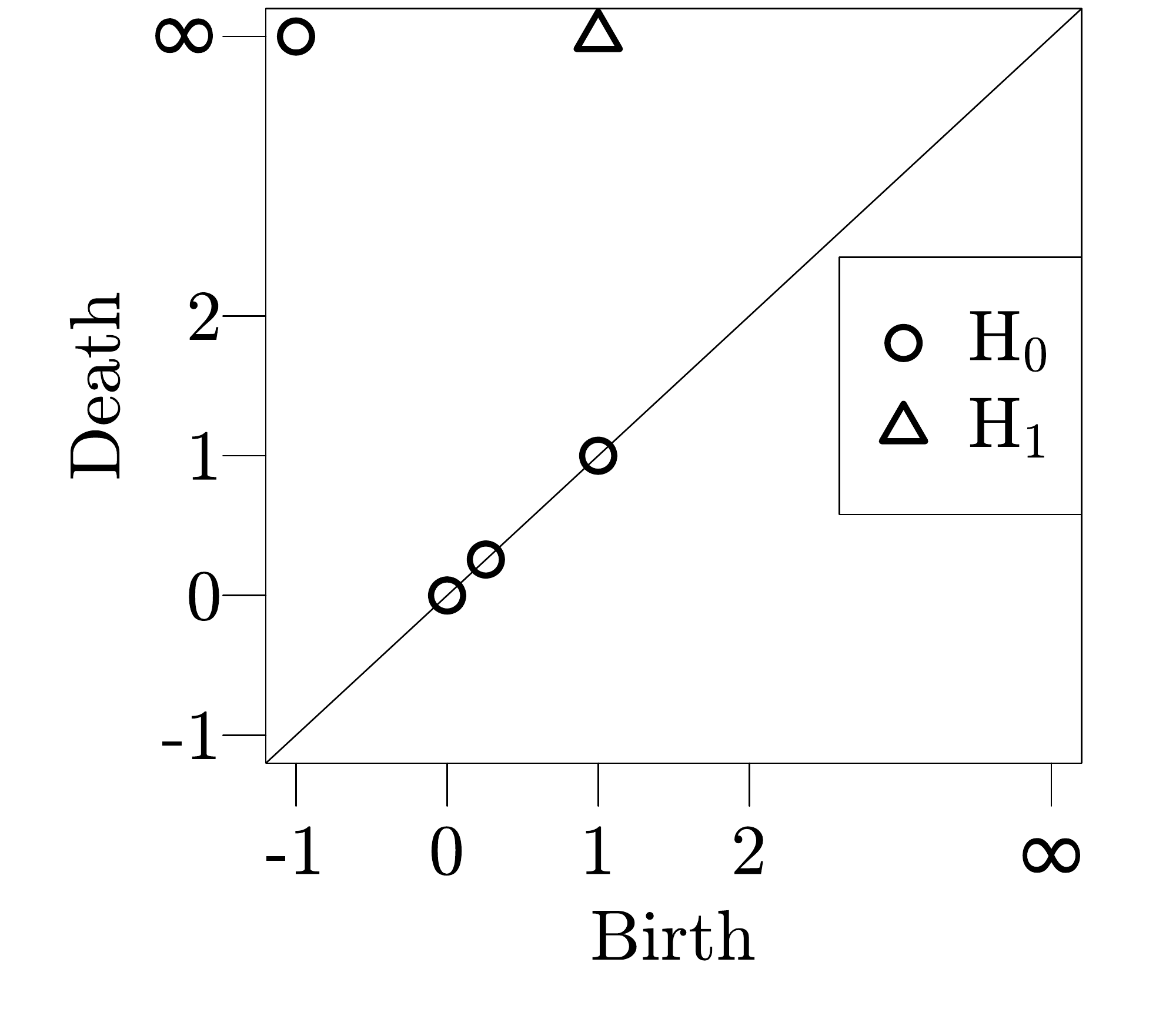}
  \caption{Diagrams for $e_1$}
\end{subfigure}%
\begin{subfigure}{.3\linewidth}
  \centering
  \includegraphics[width=\linewidth,height=.8\linewidth]{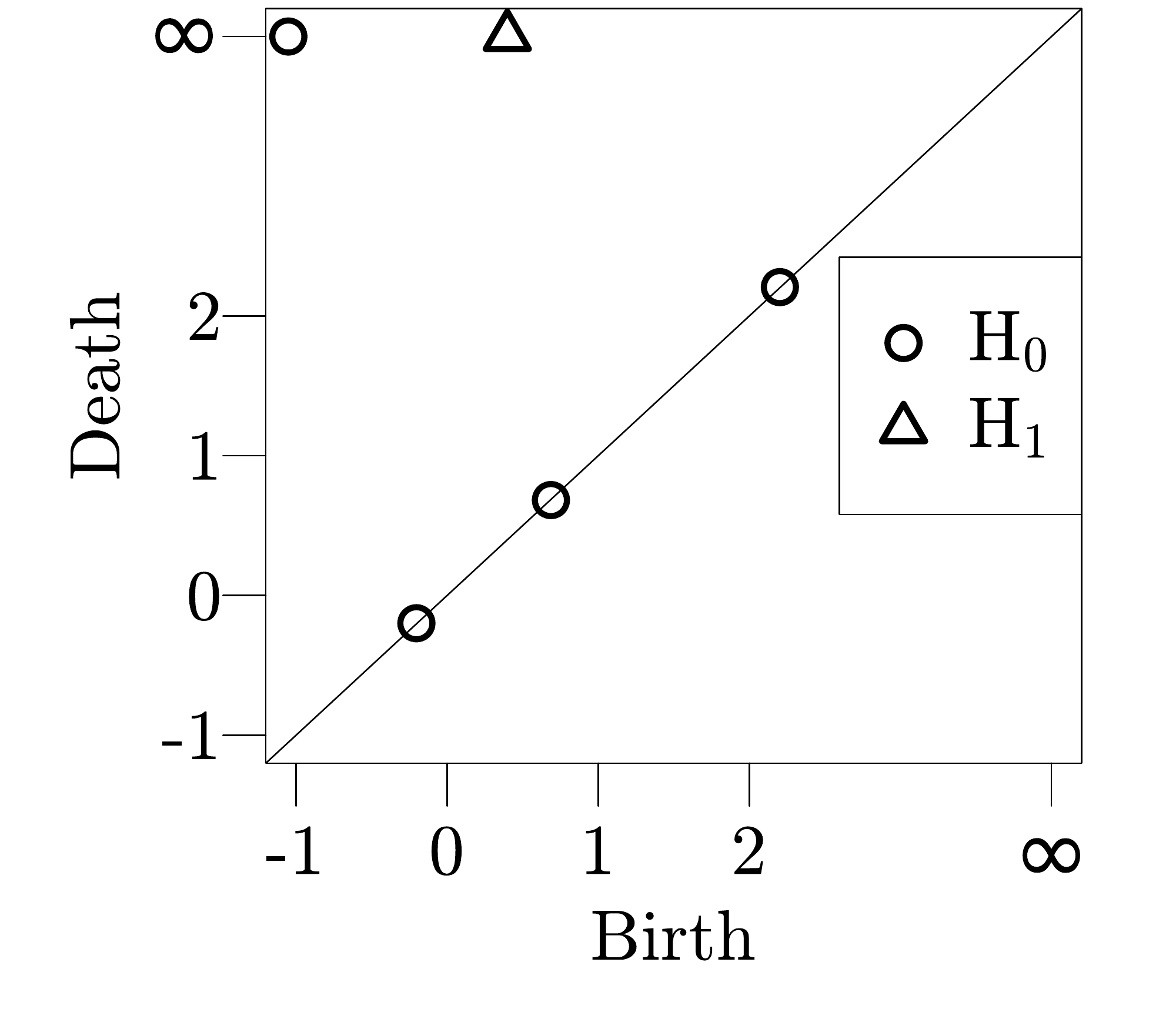}
  \caption{Diagrams for $\dir$}
\end{subfigure}
  \caption{Example of vertex reconstruction from three directions, $e_2$,
  $e_1$ and $\dir$ with corresponding persistence diagrams built for
  height filtrations from these directions. The filtration lines are the dotted
  lines superimposed over the complex.}
  \label{fig:exampleVertex}
\end{figure*}

\paragraph{Edge Reconstruction}
Next, we reconstruct all edges of $\simComp$ as described in \secref{eRec}. In order to
do so, we first compute the bow tie half-angle denoted by~$\theta$. For each
vertex~$v\in V$, we consider the cyclic ordering of the points in
$V \setminus \{v\}$ around $v$ and compute $\theta(v)$, which is the minimum angle
 between all adjacent pairs of lines through $v$. Noting that the vertex set is
$V =\{(-1,2), (0,-1), (0.25,0), (1,1)\}$, we find~$\theta(v)$ to be approximately
 ~$0.237, 0.219, 0.399$, and~$0.180$ radians, respectively. Then, we fix
$\theta$ so that \mbox{$\theta=\min_{v}\theta(v)/2=0.09$.}

Now, for each of the~$n(n-1)/2$ pairs of vertices~$(v,v') \in V^2$, we construct
a bow tie~$B$ and then use this bow tie to determine whether an edge exists between
the two vertices. We go through two examples: one for a pair of vertices that does have an
edge between them, $((0.25,0),(1,1))$, and one for a pair that does not,
$((0.25,0),(-1,2))$. First, consider the pair~$v = (0.25,0)$ and~\mbox{$v' = (1,1)$.}
To construct the bow tie at $v$ that contains $v'$, we first find a unit vector
perpendicular to the vector that points from~$v$ to~$v'$. Here, our unit vector
is~$\dir = (-0.8,0.6)$. Now, we find~$\dir_1,\dir_2$ such that the
bow tie at $v$ has half-angle~$\theta$ with~$\dir$. We choose~$\dir_1 = (-0.851,0.526)$
and~$\dir_2 = (-0.743,0.669)$. By \lemref{Indegree}, we use the persistence
diagrams from these two directions to compute~$\indeg{v}{\dir_1}$ and~$\indeg{v}{\dir_2}$.
We observe that~$\dgm{0}{\dir_1}$ contains exactly one birth-death pair~$(x,y)$
such that~$y=v \cdot \dir_1$ and $\dgm{1}{\dir_1}$ has one birth-death pair such
that~$x=v \cdot \dir_1$. Thus,~$\indeg{v}{\dir_1}=2$. On the other hand, $\dgm{0}{\dir_2}$
contains exactly one birth-death pair~$(x,y)$ such that~$y=v \cdot \dir_2$,
but~$\dgm{1}{\dir_2}$ contains no birth-death pair such that~$x=v \cdot \dir_2$.
So~$\indeg{v}{\dir_2}=1$. Now, since~$|\indeg{v}{\dir_1}-\indeg{v}{\dir_2}|=1$,
we know that~$(v,v') \in E$, by \lemref{edgeExist}.

For the second edge example, consider the pair of vertices~$v = (0.25,0)$ and~$v' = (-1,2)$.
Again, we construct the bow tie at $v$ containing $v'$, by finding a unit vector
perpendicular to the vector, $(v'-v)$. We choose this to be $s=(0.848,0.530)$.
Then, the~$\dir_1$ and~$\dir_2$ that form angle~$\theta$ with~$s$
are~\mbox{$ \dir_1 = (0.892, 0.452)$} and~$\dir_2 = (0.797,0.604)$. Again by \lemref{Indegree},
we examine the zero- and one-dimensional persistence diagrams from these two directions
to compute the indegree from each direction for vertex~$v$. In $\dgm{0}{\dir_1}$, we
have one pair~$(x,y)$ that dies at~$y=v \cdot \dir_1$, but in~$\dgm{1}{\dir_1}$, no pair
is born at~$x=v \cdot \dir_1$. So~$\indeg{v}{\dir_1}=1$. We see the exact same
for~$\dir_2$, which means that~$|\indeg{v}{\dir_1}-\indeg{v}{\dir_2}|=0$. Since
\lemref{edgeExist} tells us that we have an edge between~$v$ and~$v'$ only if the
absolute value of the difference of indegrees is one, we know that there is no edge between
vertices~$(0.25,0)$ and~$(-1,2)$.
\begin{figure*}
\centering
\begin{subfigure}{.3\linewidth}
  \centering
  \includegraphics[width=\linewidth,height=.8\linewidth]{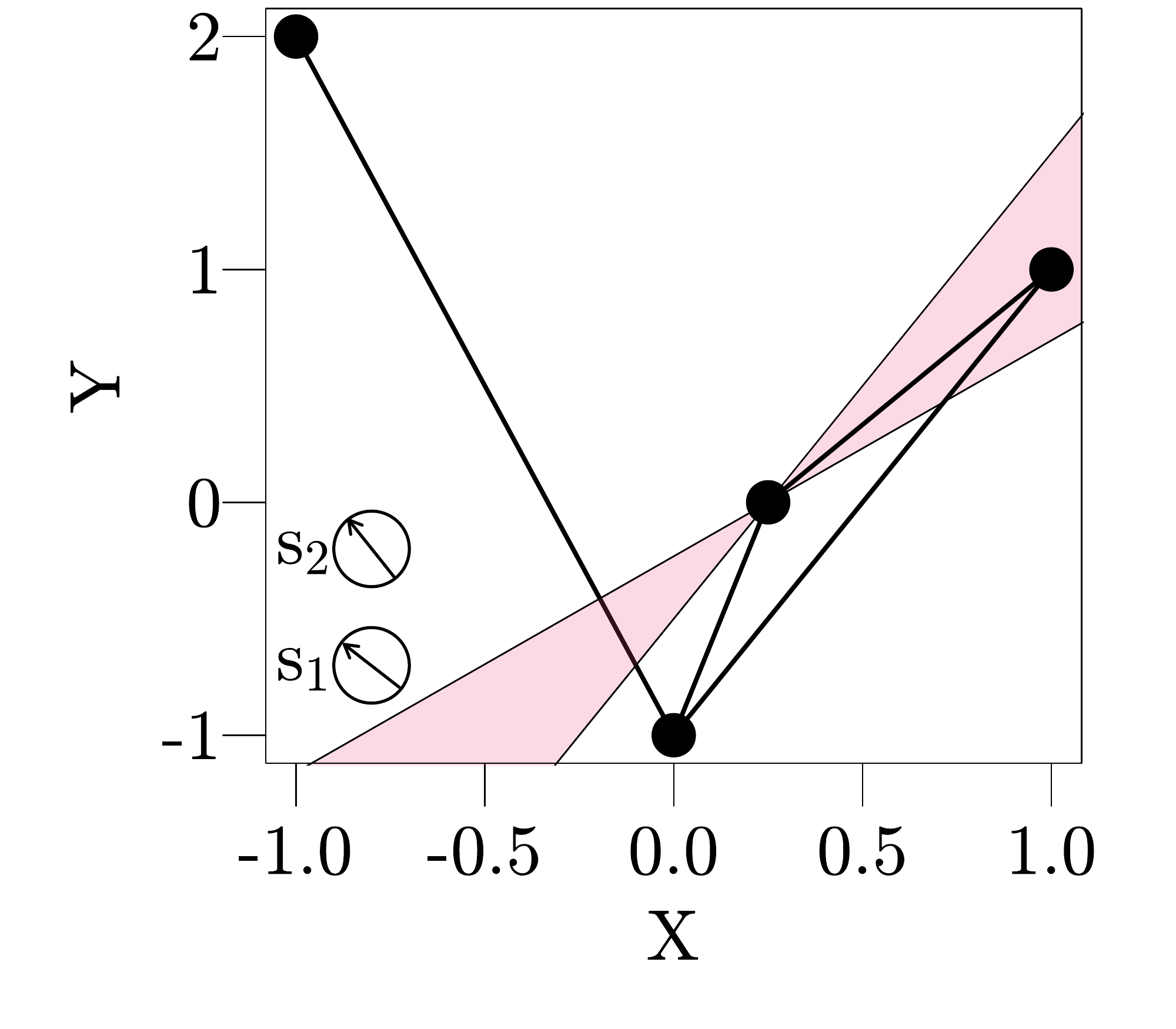}
  \caption{Bow tie lines for $\dir_1$ and $\dir_2$}
\end{subfigure}%
\begin{subfigure}{.3\linewidth}
  \centering
  \includegraphics[width=\linewidth,height=.8\linewidth]{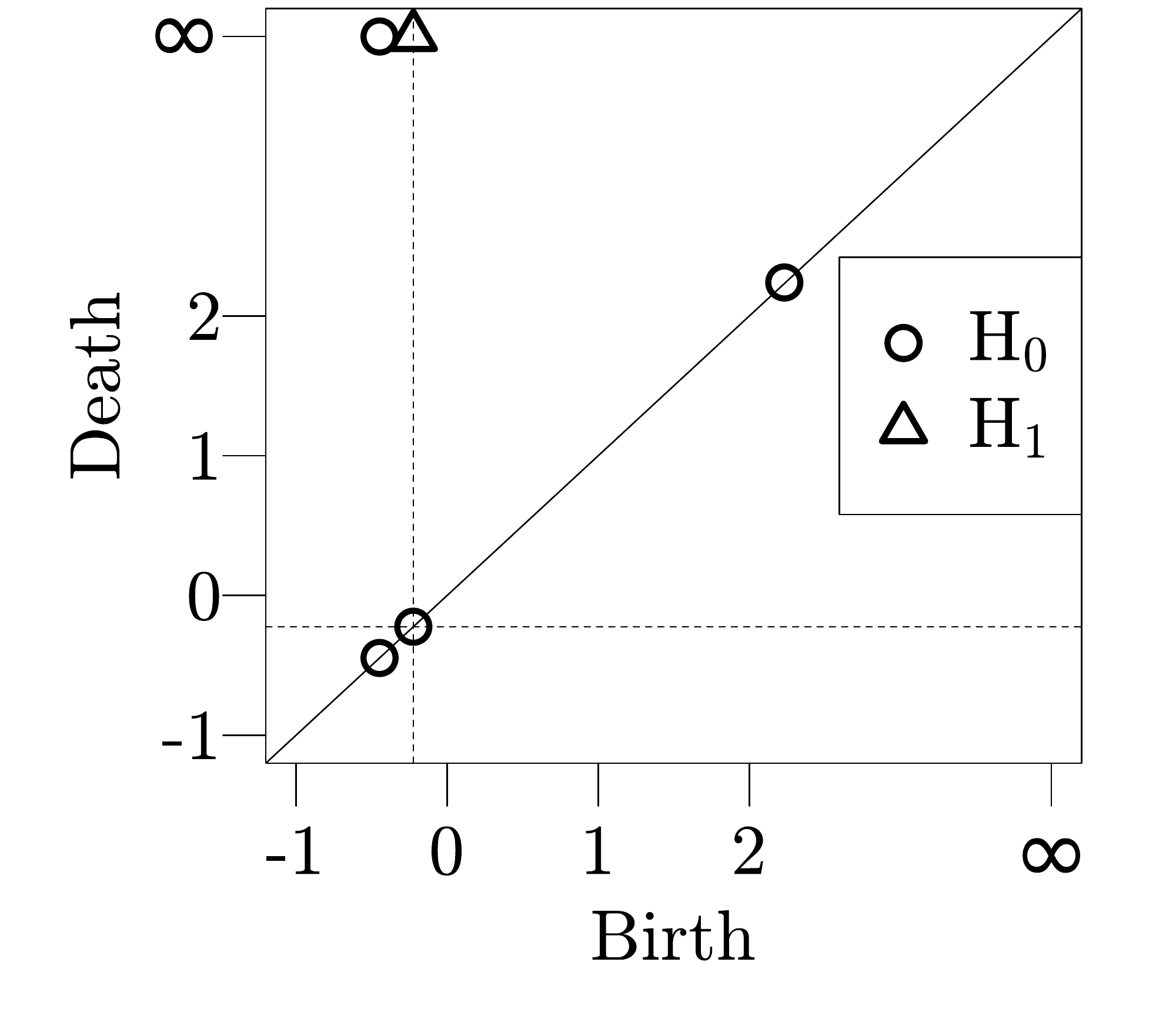}
  \caption{Diagram for $\dir_1$}
  \label{fig:Edge-second_dir}
\end{subfigure}%
\begin{subfigure}{.3\linewidth}
  \centering
  \includegraphics[width=\linewidth,height=.8\linewidth]{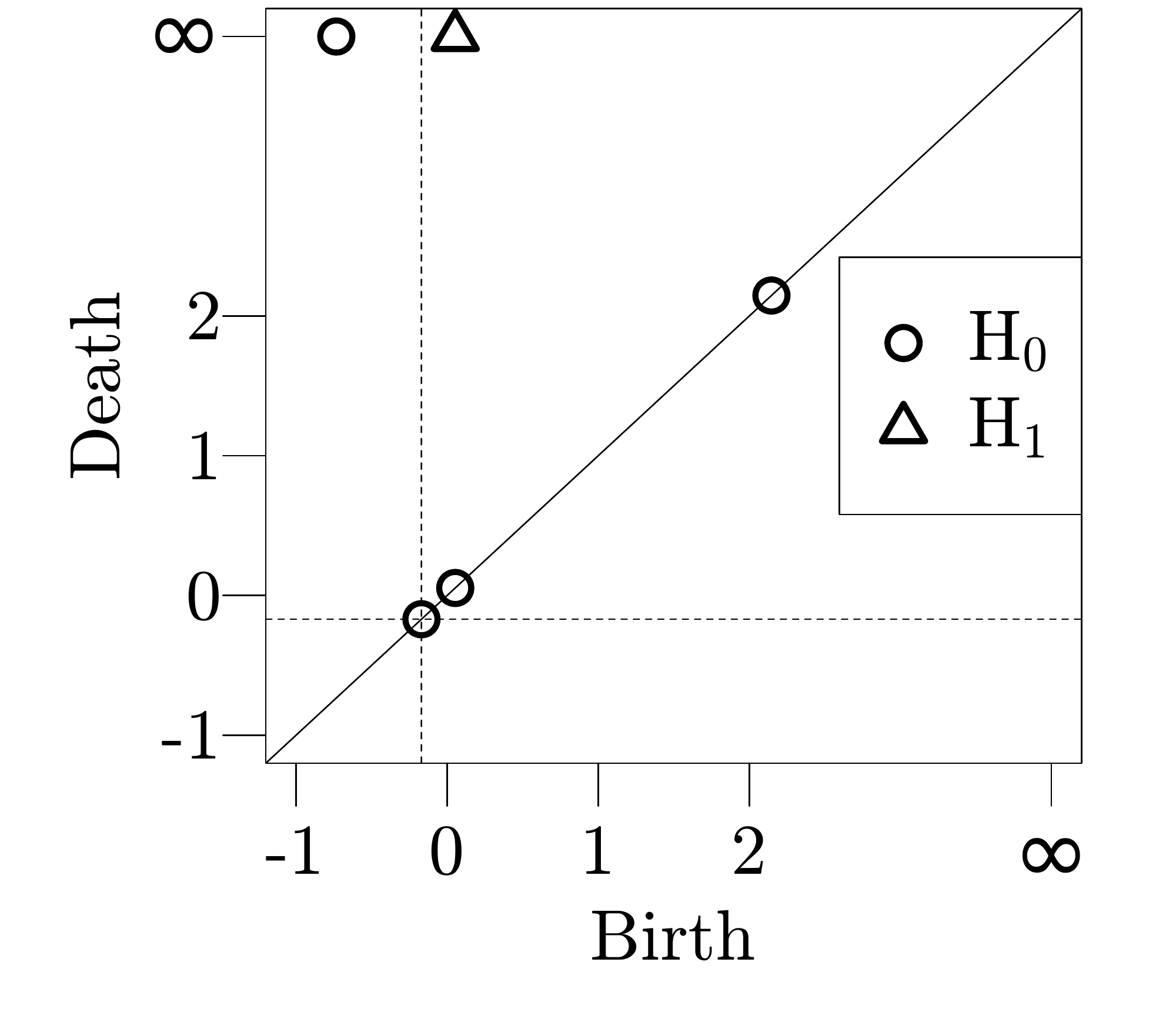}
  \caption{Diagram for $\dir_2$}
    \label{fig:Edge-third_dir}
\end{subfigure}\\
\vspace{1cm}
\begin{subfigure}{.3\linewidth}
  \centering
  \includegraphics[width=\linewidth,height=.8\linewidth]{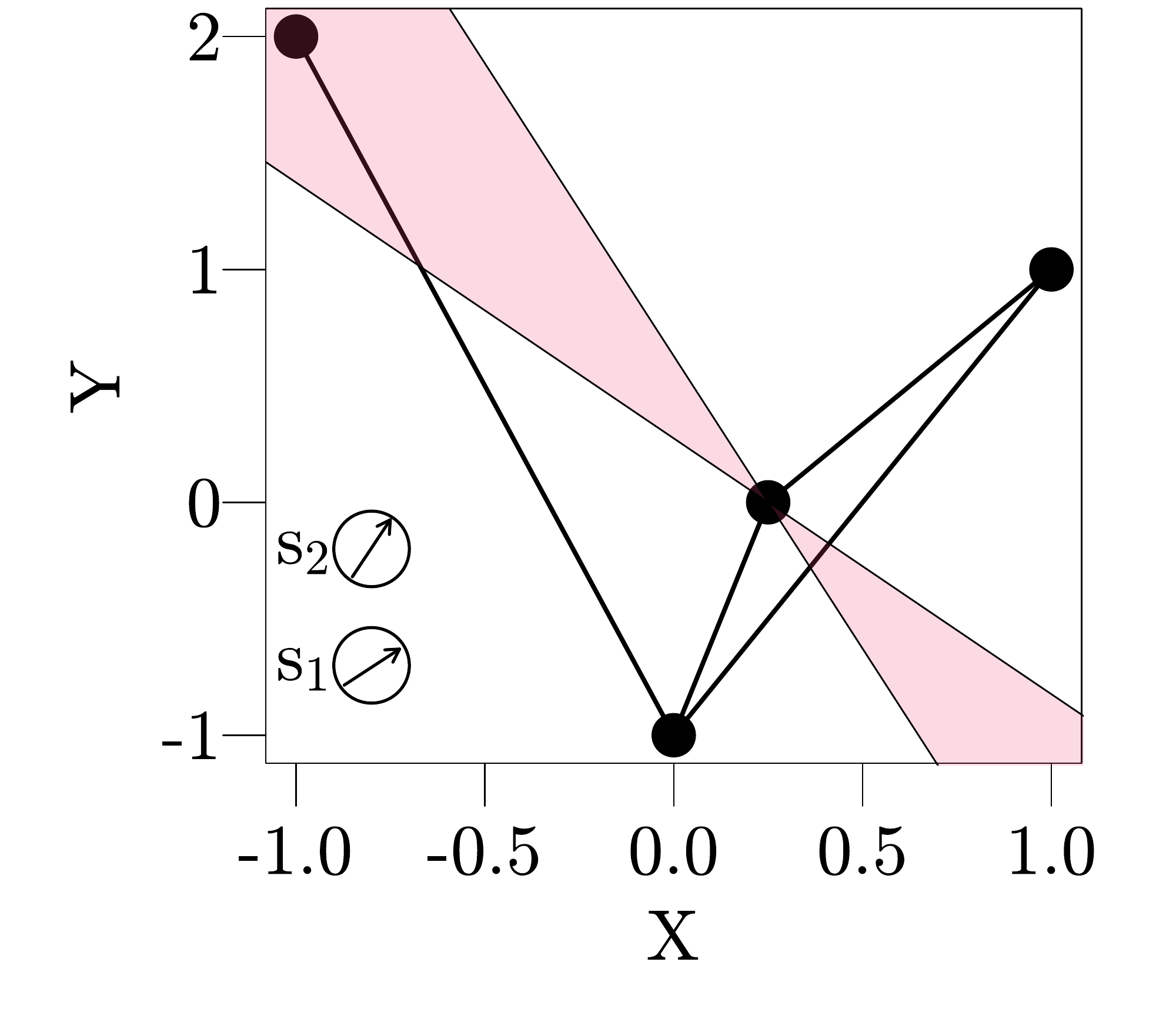}
  \caption{Bow tie lines for $\dir_1$ and $\dir_2$}
\end{subfigure}
\begin{subfigure}{.3\linewidth}
  \centering
  \includegraphics[width=\linewidth,height=.8\linewidth]{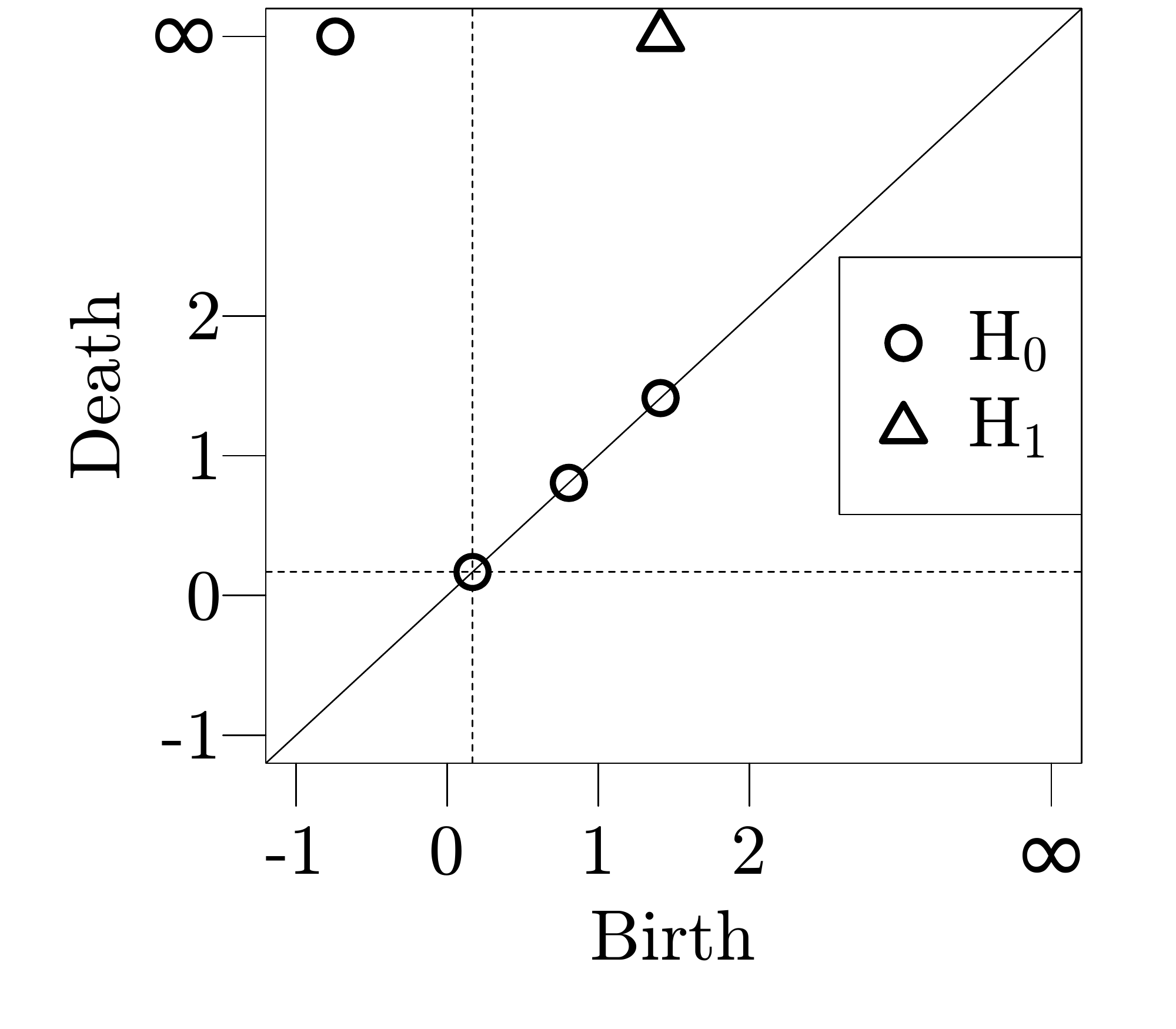}
  \caption{Diagram for $\dir_1$}
\end{subfigure}%
\begin{subfigure}{.3\linewidth}
  \centering
  \includegraphics[width=\linewidth,height=.8\linewidth]{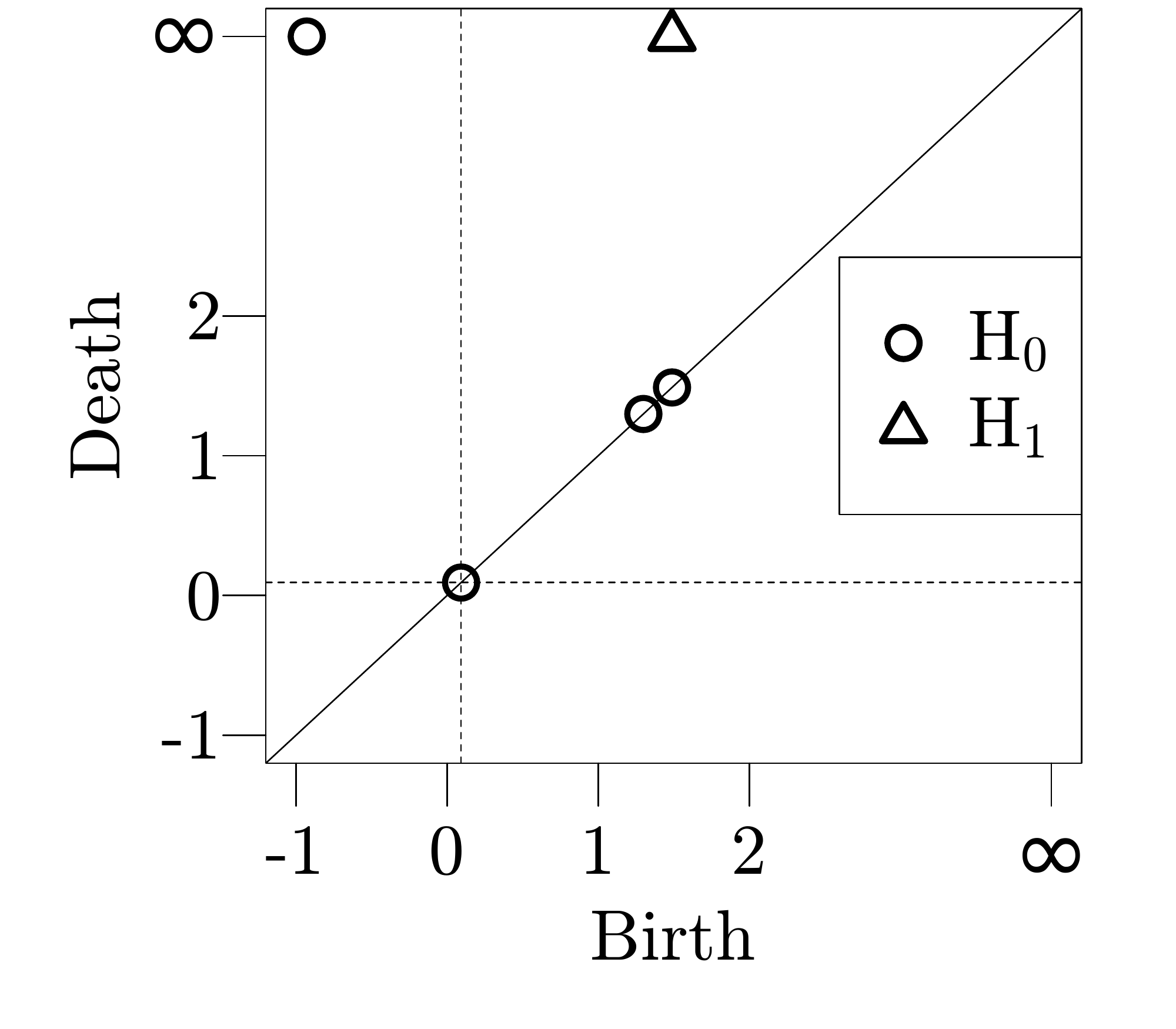}
  \caption{Diagram for $\dir_2$}
\end{subfigure}
  \caption{Example of edge reconstruction for two edges. The first edge (top
  row) exists while the second edge (bottom row) does not. The bow tie
  is given on the left while the persistence diagrams $\dgm{0}{\dir_1}$ and
  $\dgm{1}{\dir_1}$ are given in the middle and the persistence diagrams
  $\dgm{0}{\dir_2}$ and $\dgm{1}{\dir_2}$ are given on the right. The dotted lines
  indicate $v \cdot \dir_1$ and $v \cdot \dir_2$ in diagrams for $\dir_1$ and $\dir_2$ respectively.}
  \label{fig:exampleEdge}
\end{figure*}

In order to reconstruct all edges, we perform the same computations for all pairs
of vertices. After doing this, we obtain the desired plane graph as shown in~\figref{exampleVertex}.

\end{document}